\documentclass[journal]{IEEEtran}
\IEEEoverridecommandlockouts
\usepackage[noadjust]{cite}
\usepackage{amsmath,amssymb,amsfonts}
\usepackage{algorithmic}
\usepackage{graphicx}
\usepackage{url}
\usepackage{textcomp,dsfont}
\usepackage{xcolor,comment}
\usepackage{nicefrac,xfrac}
\usepackage{subcaption}
\usepackage{flushend}

\newtheorem{theorem}{Theorem}

\newtheorem{lemma}{Lemma}
\newtheorem{identity}{Identity}

\def\BibTeX{{\rm B\kern-.05em{\sc i\kern-.025em b}\kern-.08em
    T\kern-.1667em\lower.7ex\hbox{E}\kern-.125emX}}

\begin{document}

\title{Delay-Doppler Channel Estimation using \\ Arbitrarily Modulated Data Transmissions}% \\
% \thanks{
% This work is supported by the National Science Foundation under grants 2342690 and 2148212, in part by funds from federal agency and industry partners as specified in the Resilient \& Intelligent NextG Systems (RINGS) program, and in part by the Air Force Office of Scientific Research under grants FA 8750-20-2-0504 and FA 9550-23-1-0249. \\
% $*$ denotes equal contribution.}
% }

\author{Nishant Mehrotra$^*$, Sandesh Rao Mattu$^*$, Robert Calderbank~\IEEEmembership{Life Fellow,~IEEE}
        % <-this % stops a space
\thanks{This work is supported by the National Science Foundation under grants 2342690 and 2148212, in part by funds from federal agency and industry partners as specified in the Resilient \& Intelligent NextG Systems (RINGS) program, and in part by the Air Force Office of Scientific Research under grants FA 8750-20-2-0504 and FA 9550-23-1-0249. \\ The authors are with the Department of Electrical and Computer Engineering, Duke University, Durham, NC, 27708, USA (email: nishantm\allowbreak@alumni.rice.edu,~\allowbreak sandeshmattu\allowbreak@iisc.ac.in,~\allowbreak robert.calderbank\allowbreak@duke\allowbreak.edu). \\ $*$ denotes equal contribution.}% <-this % stops a space
%\thanks{Manuscript received April 19, 2021; revised August 16, 2021.}
}

% \author{\IEEEauthorblockN{Nishant Mehrotra$^*$}
% \IEEEauthorblockA{\textit{Electrical and Computer Engineering} \\
% \textit{Duke University}\\
% Durham, USA \\
% nishant.mehrotra@duke.edu\vspace{0mm}}
% \and
% \IEEEauthorblockN{Sandesh Rao Mattu$^*$}
% \IEEEauthorblockA{\textit{Electrical and Computer Engineering} \\
% \textit{Duke University}\\
% Durham, USA \\
% sandesh.mattu@duke.edu\vspace{0mm}}
% \and
% \IEEEauthorblockN{Robert Calderbank}
% \IEEEauthorblockA{\textit{Electrical and Computer Engineering} \\
% \textit{Duke University}\\
% Durham, USA \\
% robert.calderbank@duke.edu\vspace{0mm}}
% }

\maketitle
\begin{abstract}
Conventional delay-Doppler (DD) communication and sensing systems require transmitting pilot frames at every channel coherence time interval in order to keep track of channel variations at the cost of spectral efficiency. In this paper, we propose an approach to utilize data transmissions that modulate arbitrary waveforms with zero-mean, unit average energy symbols for DD channel estimation without requiring pilot transmissions in every coherence time interval. Numerical evaluation over practical doubly-selective channel models demonstrate $\sim 1.8 \times$ improvement in uncoded spectral efficiency with our proposed data-based approach over conventional pilot-based approaches across various $6$G modulation schemes.
\end{abstract}


\section{Introduction}
\label{sec:intro}

\IEEEPARstart{D}{elay}-Doppler (DD) domain signal processing is an emerging framework for next-generation wireless as networks evolve to jointly support radar sensing \& communication capabilities~\cite{NGA_ISAC_PhaseI_2025,Yuan2024_ddsurvey,otfs_book}. By processing signals in delay and Doppler, such methods provide greater resilience to double selectivity and enable forming radar images of the scattering environment by estimating the channel in the DD domain.

Conventional approaches for DD channel estimation~\cite{bitspaper1,bitspaper2,otfs_book,Mohammed2024_pulseshaping,Mattu2024_zc,Calderbank2025_interleaved,Calderbank2025_isac,Mehrotra2025_WCLSpread,Mehrotra2026_wvfcomp} require dedicated pilot symbols in every coherence time interval, limiting the achievable spectral efficiency. In this paper, we show how to reduce pilot overhead by \emph{reusing decoded data symbols} modulated on \emph{arbitrary orthonormal bases} for DD channel estimation. This increases the spectral efficiency by not requiring pilot transmissions at every coherence time interval. Fig.~\ref{fig:overview} illustrates the concept for a simple example with channel coherence time spanning two frame intervals. Fig.~\ref{fig:overview}(\subref{fig:overview1}) depicts the conventional pilot-based approach, where DD channel estimates from a pilot frame are used to detect information symbols in the subsequent data frame\footnote{While we consider separate pilot and data frames, equivalent results may also be derived using superimposed~\cite{Calderbank2025_isac} and embedded~\cite{Mohammed2024_pulseshaping,Calderbank2025_interleaved} pilot structures.}. Assuming no data symbols in the pilot frame, the achieved spectral efficiency is $\text{SE} = \nicefrac{\log_{2}{|\mathcal{A}|}}{2}$ bits/s/Hz for information symbols drawn from a constellation $\mathcal{A}$ of size $|\mathcal{A}|$. Fig.~\ref{fig:overview}(\subref{fig:overview2}) depicts the proposed data-based approach, wherein pilot frames are transmitted once only every $(F+1) > 2$ frames and DD channel estimates are obtained from decoded information symbols in each of the $F$ data frames. This increases the spectral efficiency to $\text{SE} = \nicefrac{\log_{2}{|\mathcal{A}|}}{(1+\nicefrac{1}{F})}$ bits/s/Hz.

Our scheme generalizes prior work~\cite{Mattu2025_diffdet} where, building upon concepts proposed in~\cite{Tarokh1998_diffdet1}, a similar approach was proposed specifically for the Zak-OTFS (Zak transform-based orthogonal time frequency space) modulation. In this paper, we show that the proposed data-based DD channel estimation approach is valid \emph{for any modulation scheme}, including OFDM (orthogonal frequency division multiplexing)~\cite{Ebert1971_ofdm,Bingham1990_ofdm}, AFDM (affine frequency division multiplexing)~\cite{Bemani2023_afdm,Cho2025_dftpfdma,Zhao2016_ocdm,Hanzo2025_afdm_gen}, OTSM (orthogonal time sequency division multiplexing)~\cite{Viterbo2021_otsm,Hanzo2024_otsm_amp}, and Zak-OTFS~\cite{bitspaper1,bitspaper2,otfs_book}, provided the data symbol constellation is zero-mean and has unit average energy, e.g., $M$-PSK (phase-shift keying)\footnote{PSK is also shown optimal for OFDM data-based ranging in~\cite{Liu2025_isac_data_acf,Liu2025_isac_data_survey}.}, normalized symmetric $M$-QAM (quadrature amplitude modulation), etc.

Numerical simulations with a $3$GPP-compliant Vehicular-A channel model demonstrate $\sim 1.8 \times$ improvement in uncoded spectral efficiency using the proposed data-based approach over pilot-based approaches with various modulation schemes.

\begin{table}
    \centering
    \caption{Comparison of various DD channel estimation schemes. $B$ and $T$ denote frame bandwidth and time interval, $k$ denotes the number of turbo iterations, SE: spectral efficiency.}
    {
    \setlength{\tabcolsep}{2.25pt}
    \renewcommand{\arraystretch}{1.25}
    \begin{tabular}{|c|c|c|c|}
         \hline
         Approach & Modulation & Complexity & SE \\ 
         \hline
         Separate pilot-based~\cite{bitspaper1,bitspaper2,otfs_book} & Zak-OTFS & $O(BT\log T)$ & Low \\ 
         %\hline
         Spread pilot-based~\cite{Calderbank2025_isac} & Zak-OTFS & $O(kBT\log T)$ & High \\ 
         %\hline
         Data-based~\cite{Mattu2025_diffdet} & Zak-OTFS & $O(BT\log T)$ & High \\ 
         %\hline
         \textbf{Data-based (Proposed)} & \textbf{Any} & $\mathbf{O(BT\log T)}$ & \textbf{High} \\
         \hline
    \end{tabular}
    }
    \vspace*{-0.1in}
    \label{tab:prior_work}
\end{table}

\begin{figure*}
    \centering
    \begin{subfigure}{0.49\linewidth}
    \includegraphics[width=\textwidth]{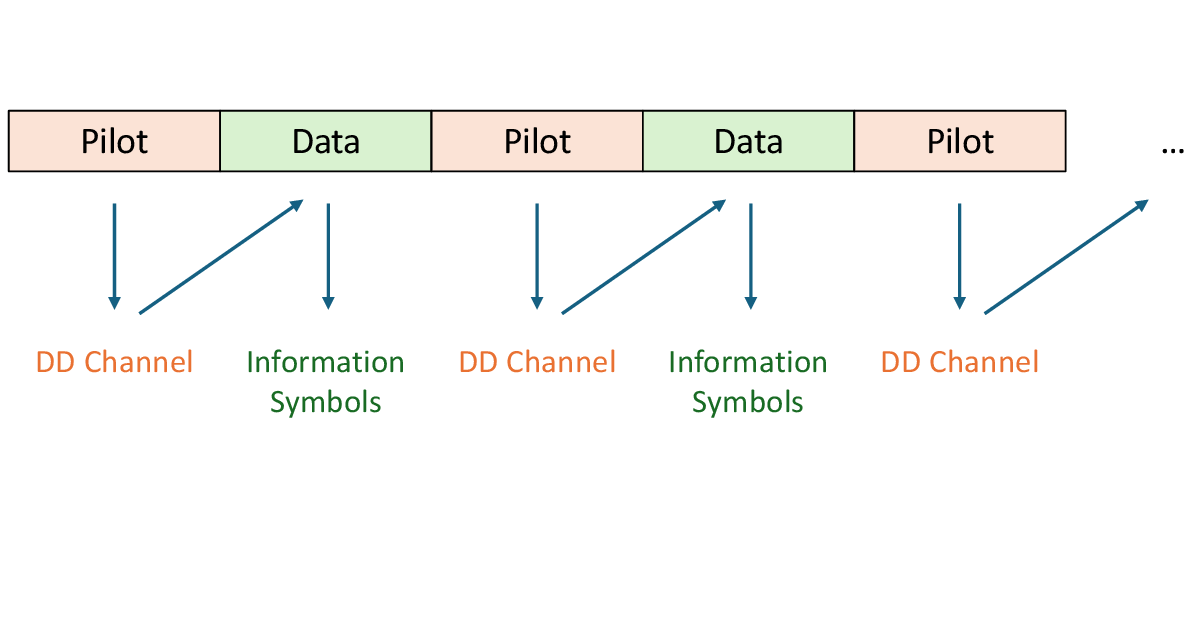}
    \caption{Conventional pilot-based DD channel estimation.}
        \label{fig:overview1}
    \end{subfigure}
    \begin{subfigure}{0.49\linewidth}
    \includegraphics[width=\textwidth]{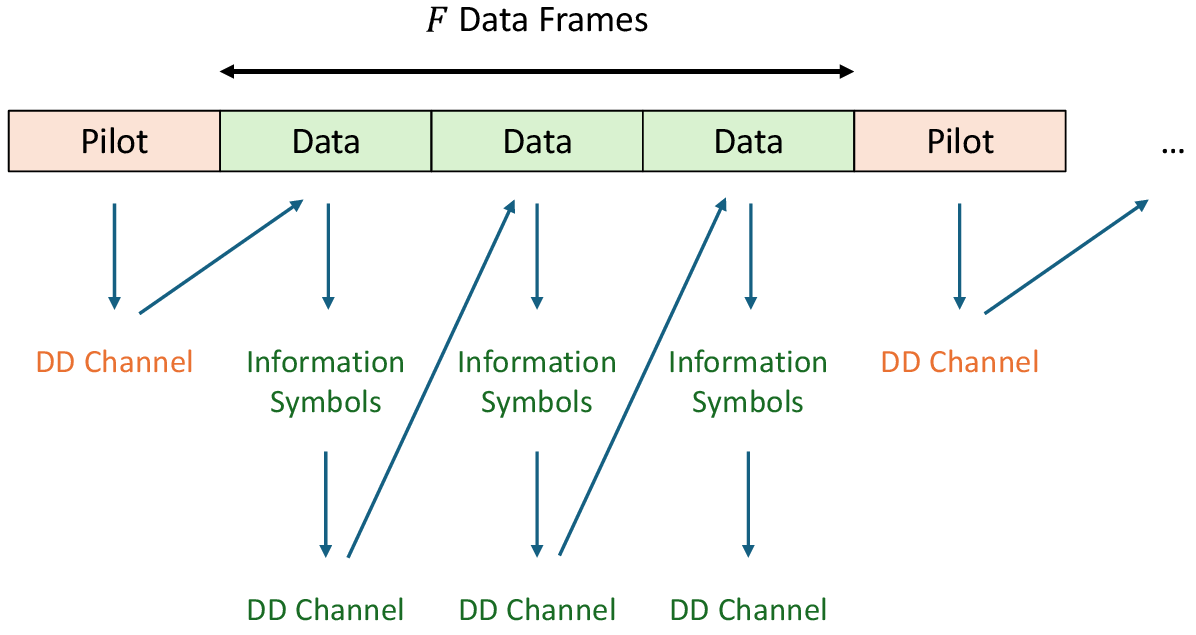}
    \caption{Proposed data-based DD channel estimation.}
        \label{fig:overview2}
    \end{subfigure}
    \caption{Frame structure assuming channel coherence time spans two frame durations, $T_{\mathsf{c}} = \frac{1}{\nu_{\max}} \approx 2T$. (a) Conventional pilot-based DD channel estimation requires transmitting pilots in every coherence interval, achieving spectral efficiency $\text{SE} = \nicefrac{\log_{2}{|\mathcal{A}|}}{2}$ bits/s/Hz. (b) The proposed approach uses data frames transmitted using any arbitrary modulation for DD channel estimation, reducing pilot overhead to once every $(F+1) > 2$ frames and achieving spectral efficiency $\text{SE} = \nicefrac{\log_{2}{|\mathcal{A}|}}{(1+\nicefrac{1}{F})}$ bits/s/Hz.}
    \label{fig:overview}
\end{figure*}

\textit{Notation:} $x$ denotes a complex scalar, $\mathbf{x}$ denotes a vector with $n$th entry $\mathbf{x}[n]$, and $\mathbf{X}$ denotes a matrix with $(n,m)$th entry $\mathbf{X}[n,m]$. $(\cdot)^{\ast}$ denotes complex conjugate, $(\cdot)^{\top}$ denotes transpose, $(\cdot)^{\mathsf{H}}$ denotes complex conjugate transpose. $\mathbb{Z}$ denotes the set of integers, $\mathbb{Z}_{N}$ the set of integers modulo $N$, and $(\cdot)_{{}_{N}}$ denotes the value modulo $N$. $\lfloor \cdot \rfloor$ and $\lceil \cdot \rceil$ denote the floor and ceiling functions. $a \odot b$ and $(a,b)$ respectively denote the bitwise dot product and greatest common divisor of two integers $a,b$. $\delta(\cdot)$ denotes the delta function, $\delta[\cdot]$ denotes the Kronecker delta function, $\mathds{1}{\{\cdot\}}$ denotes the indicator function, $\mathbf{I}_{N}$ denotes the $N \times N$ identity matrix, and $\mathbf{e}_{n}$ is the standard basis vector with value $1$ at location $n$ and zero elsewhere. %Calligraphic font $\mathcal{X}$ denotes operators or sets, with usage clear from context. $|\cdot|$ denotes set cardinality or absolute value, with usage clear from context. and . % and $(\cdot)^{-1}_{{}_{N}}$ denotes the inverse modulo $N$. and $\langle \mathbf{x}, \mathbf{y} \rangle = \sum_{n} \mathbf{x}[n] \mathbf{y}^{\ast}[n]$ denotes the inner product $\emptyset$ denotes the empty set. 

We make use of the following identity extensively.

\begin{identity}[\cite{murty2017evaluation}]
    \label{idty:sumrootsofunity}
    The sum of all $N$th roots of unity satisfies:
    \begin{align*}
        \sum_{n=0}^{N-1}e^{\frac{j2\pi}{N}kn} = \begin{cases}
        N \quad \text{if } \ k \equiv 0 \bmod{N} \\
        0 \quad \ \text{otherwise}
        \end{cases}.
    \end{align*}
\end{identity}

\section{Preliminaries: Doubly-Selective Signaling}
\label{sec:prelim}

In this Section, we describe a general discrete time system model for communication over doubly-selective channels.

The continuous time system model for communication over a doubly-selective channel is~\cite{Bello1963_ltv,bitspaper1,bitspaper2,otfs_book,Mehrotra2026_wvfcomp}:
\begin{align}
    \label{eq:prelim1}
    y(t) &= \iint h(\tau,\nu) x(t-\tau) e^{j2\pi\nu(t-\tau)} d\tau d\nu + w(t),
\end{align}
where $x(t)$ (resp. $y(t)$) denotes the transmit (resp. receive) waveform in continuous time, $w(t)$ denotes the additive noise at the receiver, and $h(\tau,\nu)$ represents the delay-Doppler (DD) channel spreading function in delay $\tau$ and Doppler $\nu$. The DD channel spreading function $h(\tau,\nu)$ captures the combined effect of transmitter pulse shaping, receiver matched filtering and propagation across the physical scattering environment with fractional delay and Doppler valued paths~\cite{otfs_book,bitspaper1,bitspaper2}:
\begin{align}
    \label{eq:prelim1a}
    \mathbf{h}(\tau,\nu) &= \mathbf{\tilde{w}}(\tau,\nu) *_\sigma \mathbf{h}_{\mathrm{phy}}(\tau,\nu) *_\sigma \mathbf{w}(\tau,\nu),
\end{align}
where $\mathbf{h}_{\mathrm{phy}}(\tau,\nu) = \sum_{i=1}^{P} h_i \delta(\tau-\tau_i) \delta(\nu-\nu_i)$ is the DD representation of the physical scattering environment with $P$ paths, $\mathbf{w}(\tau,\nu)$ is the DD transmit pulse shaping filter\footnote{Examples of DD pulse shaping filters include the sinc filter, $\mathbf{w}(\tau,\nu) = \sqrt{BT}~\text{sinc}(B\tau)~\text{sinc}(T\nu)$, Gaussian filter, root raised cosine filter, etc. See~\cite{bitspaper1,bitspaper2,otfs_book,Calderbank2025_isac,Mohammed2024_pulseshaping,Calderbank2025_interleaved,Gopalam2024_tfwindowing,Chockalingam2025_gs,Mehrotra2026_iota} for an overview of DD pulse shaping filters.}, $\mathbf{\tilde{w}}(\tau,\nu) = e^{j2\pi\nu\tau} \mathbf{w}^*(-\tau,-\nu)$ is the DD receiver matched filter, and $*_\sigma$ denotes twisted convolution\footnote{$\mathbf{a}(\tau,\nu)*_\sigma \mathbf{b}(\tau,\nu) = \iint \mathbf{a}(\tau',\nu') \mathbf{b}(\tau-\tau',\nu-\nu') e^{j2\pi\nu'(\tau-\tau')} d\tau' d\nu'$}. Let $\tau_{\max} = \max_{i\in\{1,\cdots,P\}} |\tau_i|$ and $\nu_{\max} = \max_{i\in\{1,\cdots,P\}} |\nu_i|$ respectively denote the delay and Doppler spread of the physical scattering environment.

We assume communication occurs over a finite bandwidth $B = M \Delta f$ and time interval $T = \nicefrac{N}{\Delta f}$ for integers $M,N$ and frequency spacing $\Delta f$; thus $BT = MN$ is the time-bandwidth product. Hence, we consider the discrete time version of the system model in~\eqref{eq:prelim1}, with $MN$ samples of the transmit and receive waveforms sampled at integer multiples of the delay resolution $\nicefrac{1}{B}$ and limited to duration $T$~\cite{Mehrotra2025_EURASIP,Mattu2025_npj,Mehrotra2026_wvfcomp}:
\begin{align}
    \label{eq:prelim2}
    \mathbf{y}[n] &= \sum_{k,l \in \mathbb{Z}} \mathbf{h}[k,l] \mathbf{x}[(n-k)_{{}_{MN}}] e^{\frac{j2\pi}{MN}l(n-k)} + \mathbf{w}[n],
\end{align}
where $0 \leq n \leq (MN-1)$ denotes the sampling index ($n = \lfloor MN \rfloor$ for $0 \leq t \leq T$), $\mathbf{w}$ denotes noise samples, and $\mathbf{h}[k,l] = h\big(\nicefrac{k}{B},\nicefrac{l}{T}\big)$ is the channel spreading function sampled at integer multiples of the delay / Doppler resolutions. 

Since the transmit and receive waveforms $\mathbf{x}$ and $\mathbf{y}$ in~\eqref{eq:prelim2} are $MN$-periodic sequences, $MN$ information symbols can be transmitted via an $MN$-dimensional orthonormal basis as:
\begin{align}
    \label{eq:prelim3}
    \mathbf{x}[n] &= \sum_{i = 0}^{MN-1} \mathbf{s}[i] \boldsymbol{\phi}_{i}[n],
\end{align}
where $\mathbf{s} \in \mathcal{A}^{MN \times 1}$ denotes the $MN$-length vector of information symbols drawn from a discrete constellation $\mathcal{A}$, and $\boldsymbol{\phi}$ is an orthonormal basis with $MN$ elements, with $i$th element $\boldsymbol{\phi}_{i} \in \mathbb{C}^{MN \times 1}$. Substituting~\eqref{eq:prelim3} in~\eqref{eq:prelim2}, we obtain:
\begin{align}
    \label{eq:prelim4}
    \mathbf{y}[n] = &\sum_{i = 0}^{MN-1} \mathbf{s}[i] \bigg(\underbrace{\sum_{k,l \in \mathbb{Z}} \mathbf{h}[k,l] \boldsymbol{\phi}_{i}[(n-k)_{{}_{MN}}] e^{\frac{j2\pi}{MN}l(n-k)}}_{\mathbf{G}[n,i]}\bigg) \nonumber \\ &+ \mathbf{w}[n],
\end{align}
where the term within the parenthesis denotes the $(n,i)$th element of an $MN \times MN$ channel matrix $\mathbf{G}$.

Recovering the information symbols $\mathbf{s}$ requires knowledge of the channel matrix $\mathbf{G}$, or equivalently, of $\mathbf{h}[k,l]$. Conventional approaches~\cite{bitspaper1,bitspaper2,otfs_book,Mohammed2024_pulseshaping,Mattu2024_zc,Calderbank2025_interleaved,Calderbank2025_isac,Mehrotra2025_WCLSpread,Mehrotra2026_wvfcomp} transmit known pilot sequences, typically one basis element $\mathbf{x} = \boldsymbol{\phi}_{i}$, to obtain a maximum likelihood estimate of $\mathbf{h}[k,l]$ via the \emph{cross-ambiguity function}\footnote{When $\mathbf{y} = \mathbf{x}$, $\mathbf{A}_{\mathbf{x},\mathbf{x}}[k, l]$ (abbrev. $\mathbf{A}_{\mathbf{x}}[k, l]$) is the self-ambiguity function.}: 
\begin{align}
    \label{eq:prelim6}
    \widehat{\mathbf{h}}[k,l] &= \mathbf{A}_{\mathbf{y},\mathbf{x}}[k,l] \nonumber \\
    &= \frac{1}{MN} \sum_{n = 0}^{MN-1} \mathbf{y}[n] \mathbf{x}^{*}[(n-k)_{{}_{MN}}]e^{-\frac{j2\pi}{MN}l(n-k)}.
\end{align}

Subsequently, entries of the matrix $\mathbf{G}$ are estimated via~\eqref{eq:prelim4} and used to recover the information symbols $\mathbf{s}$, e.g., via the minimum mean squared error (MMSE) estimator~\cite{otfs_book}.

Different choices of the basis $\boldsymbol{\phi}$ result in different modulation schemes\footnote{See~\cite{Mehrotra2026_wvfcomp} for a unified discussion on various bases / modulation schemes.}, e.g., OFDM, AFDM, OTSM and Zak-OTFS. %OFDM (orthogonal frequency division multiplexing)~\cite{Ebert1971_ofdm,Bingham1990_ofdm}, AFDM (affine frequency division multiplexing)~\cite{Bemani2023_afdm,Cho2025_dftpfdma,Zhao2016_ocdm,Hanzo2025_afdm_gen}, OTSM (orthogonal time sequency division multiplexing)~\cite{Viterbo2021_otsm,Hanzo2024_otsm_amp}, and Zak-OTFS (orthogonal time frequency space modulation)~\cite{bitspaper1,bitspaper2,otfs_book}.

\subsubsection{OFDM}
\label{subsubsec:prelim_ofdm}

The basis element in OFDM is~\cite{Ebert1971_ofdm,Bingham1990_ofdm}:
\begin{align}
    \label{eq:ofdm1}
    \boldsymbol{\phi}_{i}[n] = \frac{1}{\sqrt{M}} e^{\frac{j2\pi}{M}in} \mathds{1}\big\{\lfloor\nicefrac{n}{M}\rfloor = \lfloor\nicefrac{i}{M}\rfloor\big\}.
\end{align}

\subsubsection{AFDM}
\label{subsubsec:prelim_afdm}

The basis element in AFDM is~\cite{Bemani2023_afdm,Cho2025_dftpfdma,Hanzo2025_afdm_gen}:
\begin{align}
    \label{eq:afdm1}
    \boldsymbol{\phi}_{i}[n] = \frac{1}{\sqrt{MN}}e^{j2\pi\big(c_1n^2+c_2i^2+\frac{ni}{MN}\big)},
\end{align}
where $c_1, c_2 \in \mathbb{Z}$. The AFDM basis specializes to OCDM~\cite{Bemani2023_afdm} when $c_1 = c_2 = \nicefrac{1}{2MN}$ and to DFT-p-FDMA~\cite{Cho2025_dftpfdma} when $c_1 = c_2 = \nicefrac{\Delta}{MN}$, where $(\Delta,MN) = 1$.

% \subsubsection{ODDM}
% \label{subsubsec:prelim_oddm}

% The basis element in ODDM is~\cite{Yuan2022_oddm,Yuan2024_oddm_gen}:
% \begin{align}
%     \label{eq:oddm1}
%     \boldsymbol{\phi}_{i}[n] = \frac{1}{\sqrt{N}} e^{\frac{j2\pi}{N}{\lfloor\nicefrac{i}{M}\rfloor} \lfloor\nicefrac{n}{M}\rfloor} \mathds{1}\big\{n \equiv i \bmod{M}\big\}.
% \end{align}

\subsubsection{OTSM}
\label{subsubsec:prelim_otsm}

The basis element in OTSM is~\cite{Viterbo2021_otsm,Hanzo2024_otsm_amp}:
\begin{align}
    \label{eq:otsm1}
    \boldsymbol{\phi}_{i}[n] = \frac{1}{\sqrt{N}} (-1)^{\lfloor\nicefrac{i}{M}\rfloor \odot \lfloor\nicefrac{n}{M}\rfloor} \mathds{1}\big\{n \equiv i \bmod{M}\big\},
\end{align}
where $\odot$ denotes the bitwise dot product.

\subsubsection{Zak-OTFS}
\label{subsubsec:prelim_zak}

The basis element in Zak-OTFS is~\cite{bitspaper1,bitspaper2,otfs_book}:
\begin{align}
    \label{eq:zakotfs1}
    \boldsymbol{\phi}_{i}[n] &= \frac{1}{\sqrt{N}} \sum_{d \in \mathbb{Z}} e^{j\frac{2\pi}{N} d \lfloor \nicefrac{i}{M} \rfloor} \delta[n-(i)_{{}_{M}}-dM] \nonumber \\
    &= \frac{1}{\sqrt{N}} e^{\frac{j2\pi}{N}{\lfloor\nicefrac{i}{M}\rfloor} \lfloor\nicefrac{n}{M}\rfloor} \mathds{1}\big\{n \equiv i \bmod{M}\big\},
\end{align}
termed \emph{pulsone} (pulse train modulated by a tone).

\section{Data-Based DD Channel Estimation}
\label{sec:diff_comm}

As mentioned in Section~\ref{sec:prelim}, conventional approaches estimate the DD channel spreading function $\mathbf{h}[k,l]$ by transmitting known pilot symbols in every coherence time interval, limiting the achievable throughput. In this Section, we establish that DD channel estimation is possible using \emph{data symbols} modulated on \emph{any arbitrary orthonormal basis} $\boldsymbol{\phi}$, including the bases described in Section~\ref{sec:prelim}. This property enables \emph{reusing decoded data as pilots} to estimate the DD channel without requiring pilot transmissions in every coherence time interval.

Section~\ref{subsec:diff_comm_approach} describes the core innovation underlying our approach. Assuming known information symbols, in Theorem~\ref{thm:data_amb} we establish that data-based DD channel estimates match the ground-truth channel spreading function $\mathbf{h}[k,l]$ in expectation for arbitrary orthonormal bases $\boldsymbol{\phi}$ under benign conditions on the information symbol constellation $\mathcal{A}$. In Section~\ref{subsec:diff_comm_feasibility}, we derive necessary conditions for practical implementation of data-based DD channel estimation.

\subsection{DD Channel Estimation using Data Frames}
\label{subsec:diff_comm_approach}

Prior to deriving our main result in Theorem~\ref{thm:data_amb}, we present the following Lemma on the cross-ambiguity-based estimate of the sampled DD channel spreading function $\mathbf{h}[k,l]$ via~\eqref{eq:prelim6}.

\begin{lemma}
    \label{lmm:heff_twistconv}
    In the absence of noise\footnote{In the presence of noise, $\widehat{\mathbf{h}}[k,l]$ is limited to the coarsely known channel support $\mathcal{S}$ (with $\geq 2$ guard bins to account for pulse shape spread). The noise standard deviation is calculated in the complement set $\mathcal{S}^{c}$ and $\widehat{\mathbf{h}}[k,l]$ is subsequently thresholded in magnitude to $3 \times$ the noise standard deviation~\cite{Calderbank2025_isac}.}, the cross-ambiguity-based estimate $\widehat{\mathbf{h}}[k,l]$ from~\eqref{eq:prelim6} is given by the discrete twisted convolution of $\mathbf{h}[k,l]$ and the self-ambiguity function of $\mathbf{x}$:
    \begin{align*}
        \widehat{\mathbf{h}}[k,l]\!&= \mathbf{h}[k,l] *_{\sigma_{{}_{d}}} \mathbf{A}_{\mathbf{x}}[k, l] \nonumber \\
        &=\!\sum_{k',l' \in \mathbb{Z}} \mathbf{h}[k',l'] \mathbf{A}_{\mathbf{x}}[(k-k'), (l-l')] e^{\frac{j2\pi}{MN}l'(k-k')}.
    \end{align*}
\end{lemma}

\begin{IEEEproof}
    Substituting~\eqref{eq:prelim4} in the absence of noise into~\eqref{eq:prelim6}:
    \begin{align*}
        % \label{eq:radar6}
        \widehat{\mathbf{h}}[k,l]\!&=\!\frac{1}{MN}\!\sum_{n=0}^{MN-1} \mathbf{y}[n] \mathbf{x}^{*}[(n-k)_{{}_{MN}}]e^{-\frac{j2\pi}{MN}l(n-k)} \nonumber \\
        % &=\!\sum_{k',l' \in \mathbb{Z}} \mathbf{h}[k',l'] \sum_{n=0}^{MN-1} \mathbf{x}[(n-k')_{{}_{MN}}] \mathbf{x}^{*}[(n-k)_{{}_{MN}}] \nonumber \\ &\qquad\qquad\qquad\qquad\quad\times \frac{1}{MN} e^{\frac{j2\pi}{MN}\big[l'(n-k')-l(n-k)\big]} \nonumber \\
        &=\!\sum_{k',l' \in \mathbb{Z}} \mathbf{h}[k',l'] \sum_{n'=0}^{MN-1} \mathbf{x}[n'] \mathbf{x}^{*}[(n'-(k-k'))_{{}_{MN}}] \nonumber \\ &\qquad\qquad\qquad\qquad\quad\times \frac{1}{MN} e^{\frac{j2\pi}{MN}\big[(l'-l)n'+l(k-k')\big]} \nonumber \\
        &=\!\sum_{k',l' \in \mathbb{Z}} \mathbf{h}[k',l'] \mathbf{A}_{\mathbf{x}}[(k-k'), (l-l')] e^{\frac{j2\pi}{MN}l'(k-k')}.
    \end{align*}
\end{IEEEproof}

In other words, Lemma~\ref{lmm:heff_twistconv} shows that the estimate $\widehat{\mathbf{h}}[k,l]$ corresponds to the ground truth DD channel spreading function $\mathbf{h}[k,l]$ ``blurred'' by the self-ambiguity function $\mathbf{A}_{\mathbf{x}}[k,l]$. Ideally, for a ``thumbtack'' self-ambiguity function, $\mathbf{A}_{\mathbf{x}}[k,l] = \delta[k] \delta[l]$, $\widehat{\mathbf{h}}[k,l] = \mathbf{h}[k,l]$. Since an ideal ``thumbtack'' self-ambiguity function cannot be realized due to Moyal's Identity~\cite{Mehrotra2025_EURASIP}, prior works~\cite{bitspaper1,bitspaper2,otfs_book,Mohammed2024_pulseshaping,Mattu2024_zc,Calderbank2025_interleaved,Calderbank2025_isac,Mehrotra2025_WCLSpread,Mehrotra2026_wvfcomp} use basis elements of \emph{predictable} bases as pilot sequences, $\mathbf{x} = \boldsymbol{\phi}_{i}$, that have \emph{local} ``thumbtack'' self-ambiguity~\cite{Mehrotra2026_wvfcomp}: $\mathbf{A}_{\boldsymbol{\phi}_{i}}[0,0] = 1,\mathbf{A}_{\boldsymbol{\phi}_{i}}[k',l'] = 0~\text{for}~(k',l')\in\mathcal{K}_{\mathcal{S}}\setminus\{(0,0)\}$, where for DD channel support $\mathcal{S} = \big\{(k,l):|\mathbf{h}[k,l]|\neq0\big\}$, $\mathcal{K}_{\mathcal{S}} = \big\{(k',l'):\mathcal{S}\cap\big(\mathcal{S}+(k',l')\big)\neq\emptyset\big\}$. When the channel is supported within $\mathcal{S}$, this property ensures $\widehat{\mathbf{h}}[k,l] = \mathbf{h}[k,l]$ for all $k,l \in \mathcal{S}$. Examples of bases satisfying this property include AFDM, OTSM and Zak-OTFS, but not OFDM~\cite{Mehrotra2026_wvfcomp}. However, a drawback of this approach is that it requires transmitting a pilot frame at every coherence time interval, which limits the spectral efficiency.

In the following Theorem, we show that the self-ambiguity function $\mathbf{A}_{\mathbf{x}}[k,l]$ of data-modulated waveforms $\mathbf{x}$ as in~\eqref{eq:prelim3} approximates an ideal ``thumbtack'' in expected value, regardless of the choice of basis $\boldsymbol{\phi}$. This property enables using \emph{data frames} transmitted using \emph{any arbitrary modulation scheme} -- including OFDM which does not satisfy the aforementioned predictability condition for pilot-based DD channel estimation -- to estimate the DD channel without requiring pilot transmissions in every coherence time interval. Fig.~\ref{fig:overview} illustrates our high-level approach for an example with channel coherence time spanning two frame durations. Our proposed data-based channel estimation approach reduces the pilot overhead from $\nicefrac{1}{2}$ to $\nicefrac{1}{(F+1)}$, for $(F+1) > 2$, thus increasing the spectral efficiency from $\nicefrac{\log_{2}{|\mathcal{A}|}}{2}$ bits/s/Hz to $\nicefrac{\log_{2}{|\mathcal{A}|}}{(1+\nicefrac{1}{F})}$ bits/s/Hz.

\begin{theorem}
    \label{thm:data_amb}
    For data-modulated waveforms $\mathbf{x}$ from~\eqref{eq:prelim3}, any orthonormal basis $\boldsymbol{\phi}$ modulated using zero-mean, unit average energy constellations $\mathcal{A}$ has a self-ambiguity function approximating an ideal ``thumbtack'' in expectation, $\mathbb{E}\big[\mathbf{A}_{\mathbf{x}}[k,l]\big] = \mathds{1}\big\{k,l \equiv 0 \bmod{MN}\big\}$, for which the cross-ambiguity function in~\eqref{eq:prelim6} is an unbiased estimator, $\mathbb{E}\big[\widehat{\mathbf{h}}[k,l]\big] = \mathbf{h}[k,l]$.
\end{theorem}

\begin{figure*}
    \centering
    \begin{subfigure}{0.24\linewidth}
    \includegraphics[width=\textwidth]{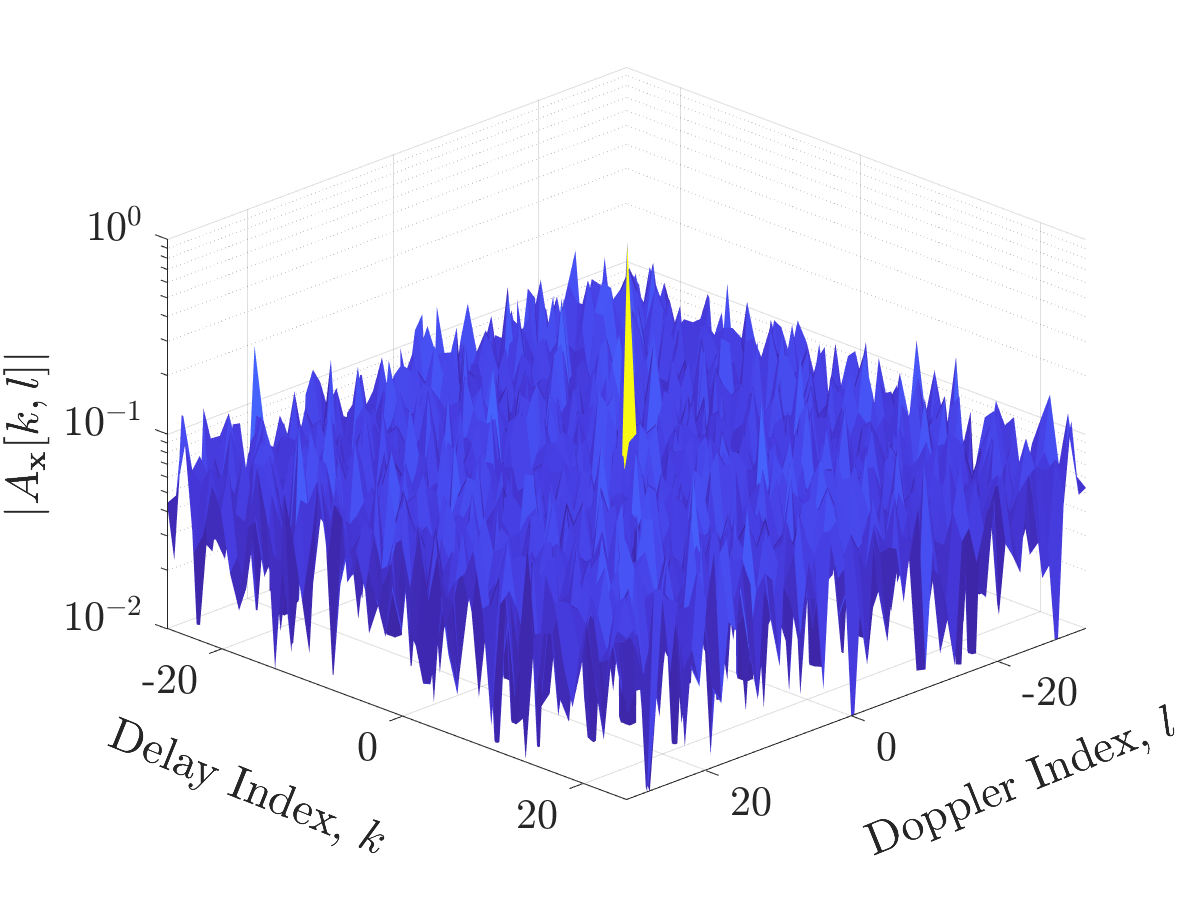}
    \caption{Zak-OTFS}
        \label{fig:crossamb_data1}
    \end{subfigure}
    \begin{subfigure}{0.24\linewidth}
    \includegraphics[width=\textwidth]{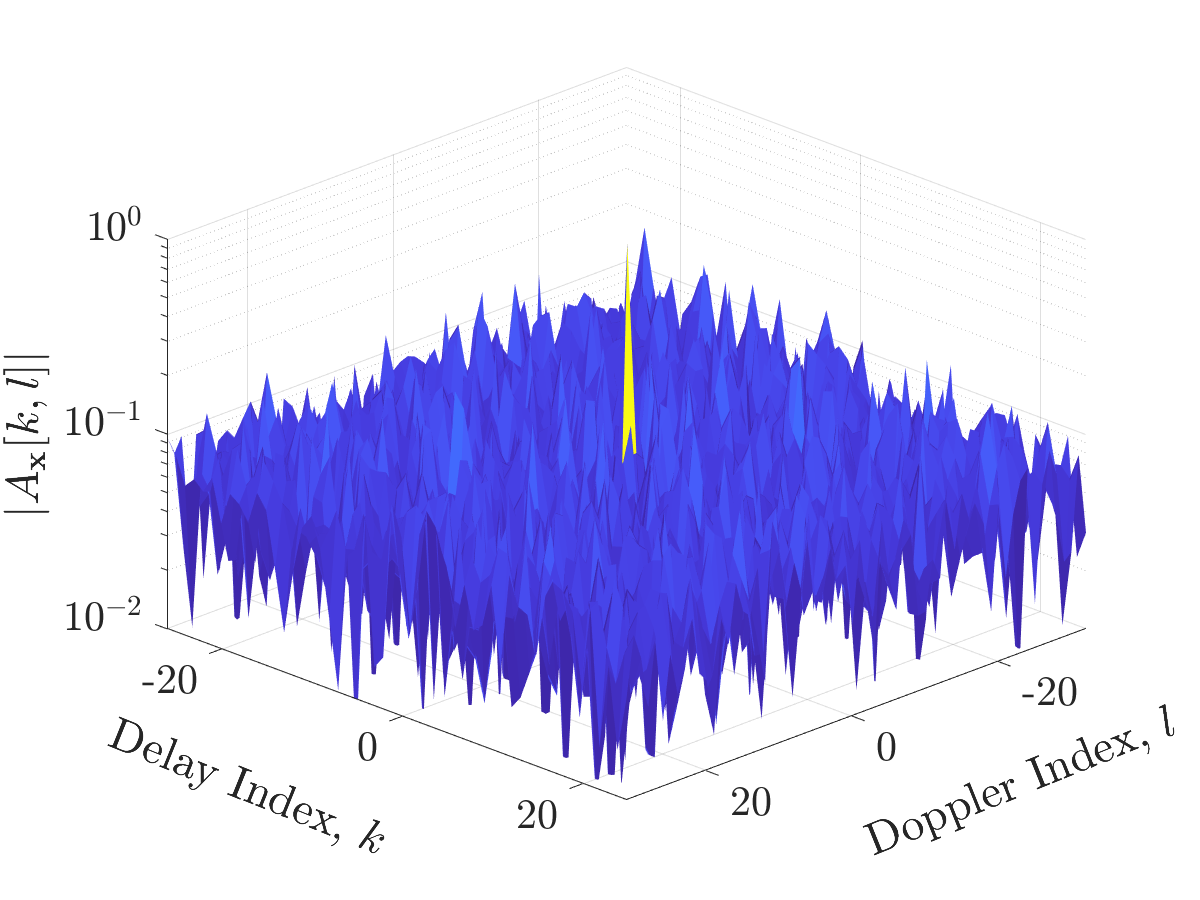}
    \caption{OTSM}
        \label{fig:crossamb_data2}
    \end{subfigure}
    \begin{subfigure}{0.24\linewidth}
    \includegraphics[width=\textwidth]{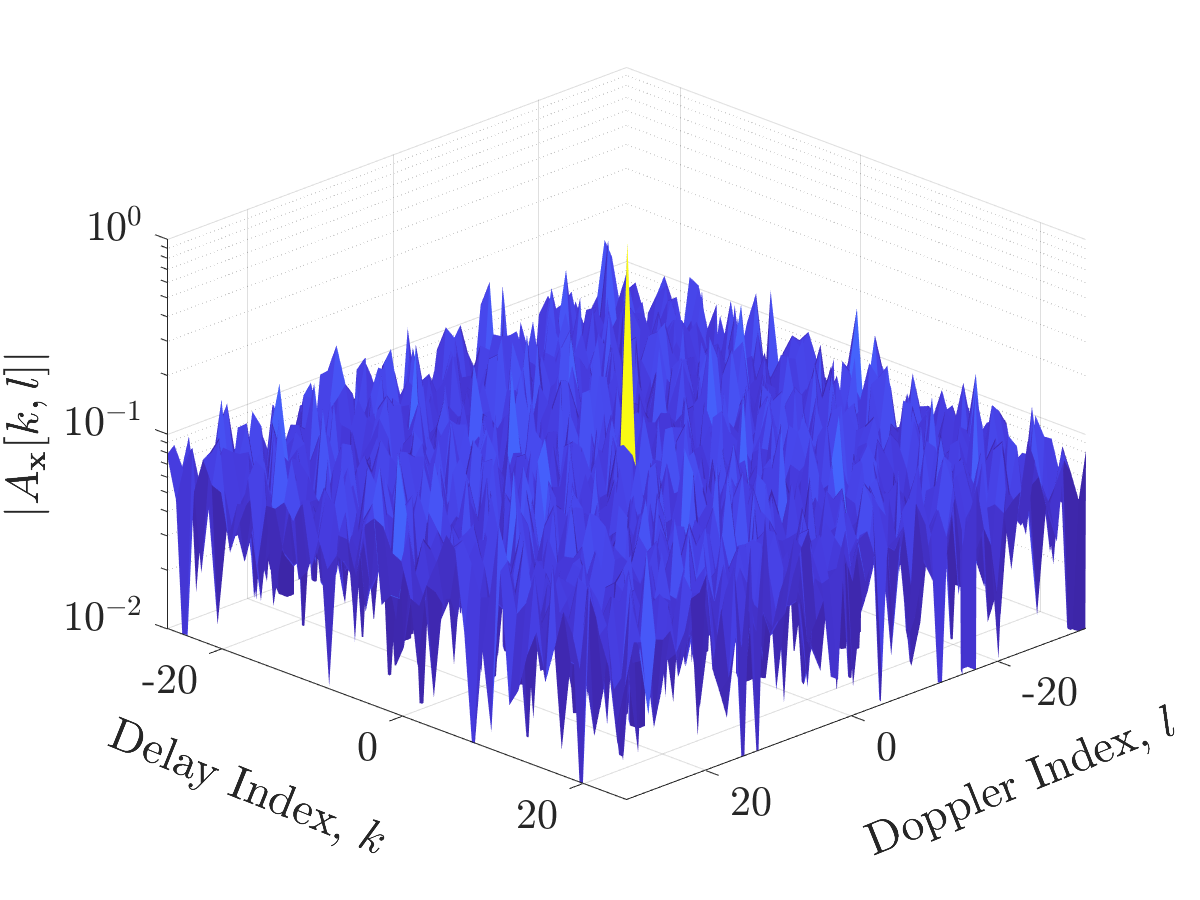}
    \caption{AFDM}
        \label{fig:crossamb_data3}
    \end{subfigure}
    \begin{subfigure}{0.24\linewidth}
    \includegraphics[width=\textwidth]{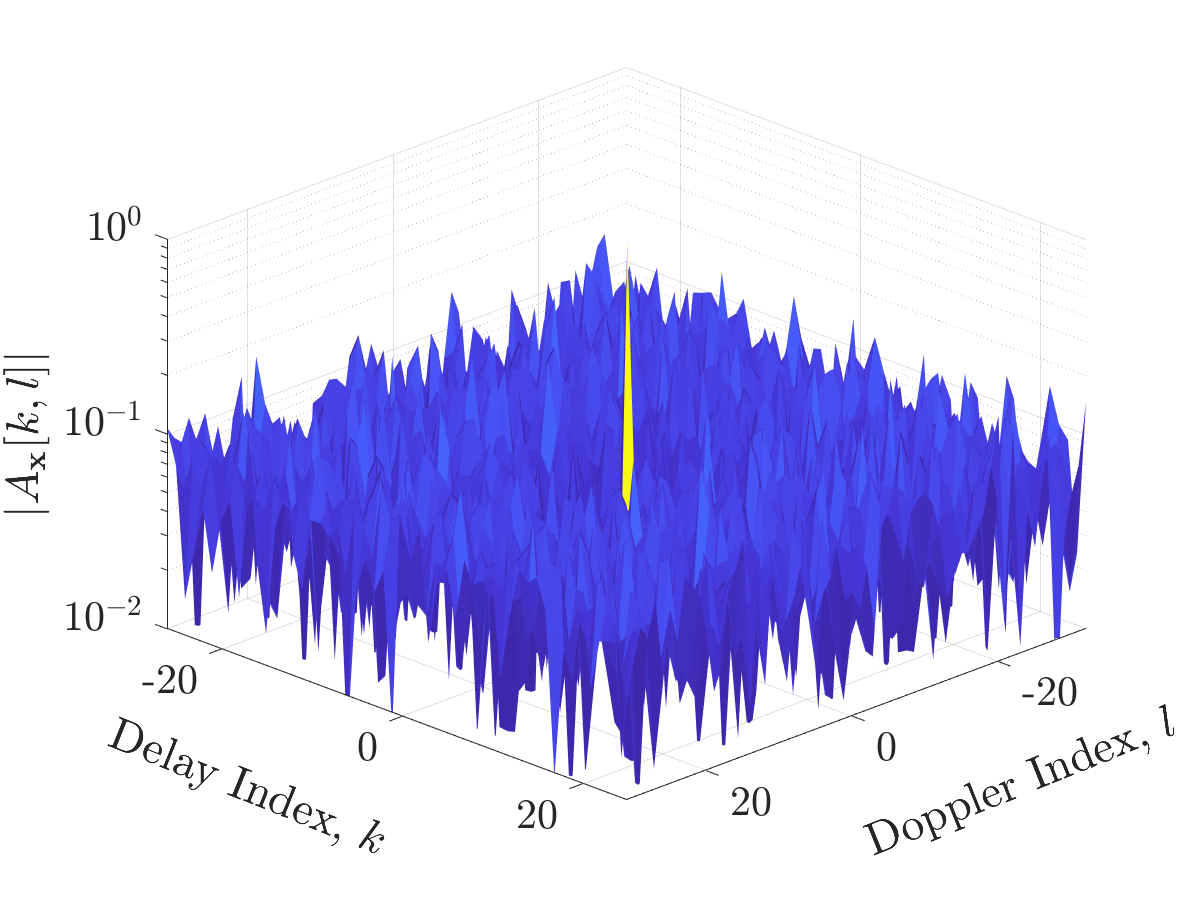}
    \caption{OFDM}
        \label{fig:crossamb_data4}
    \end{subfigure}
    \hfill
    \begin{subfigure}{0.24\linewidth}
    \includegraphics[width=\textwidth]{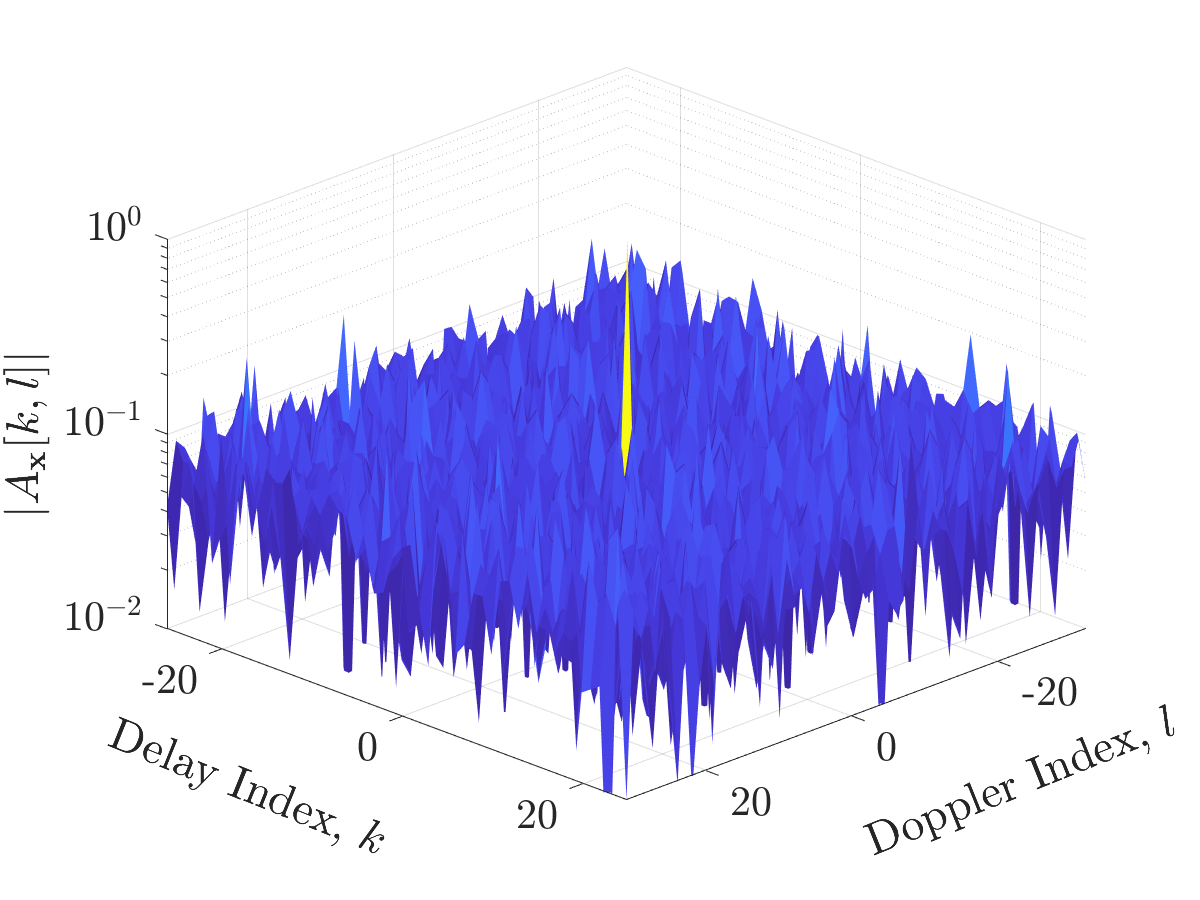}
    \caption{Zak-OTFS}
        \label{fig:crossamb_data1}
    \end{subfigure}
    \begin{subfigure}{0.24\linewidth}
    \includegraphics[width=\textwidth]{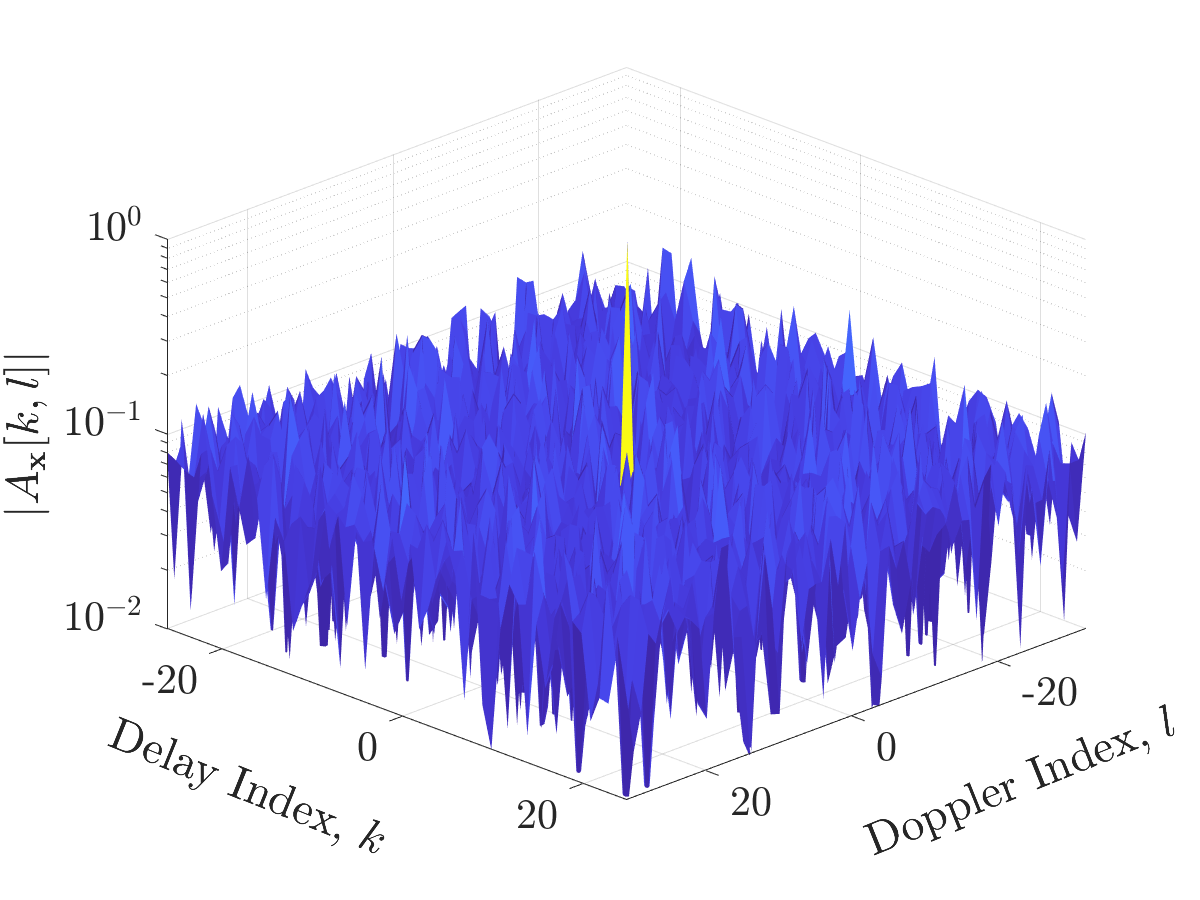}
    \caption{OTSM}
        \label{fig:crossamb_data2}
    \end{subfigure}
    \begin{subfigure}{0.24\linewidth}
    \includegraphics[width=\textwidth]{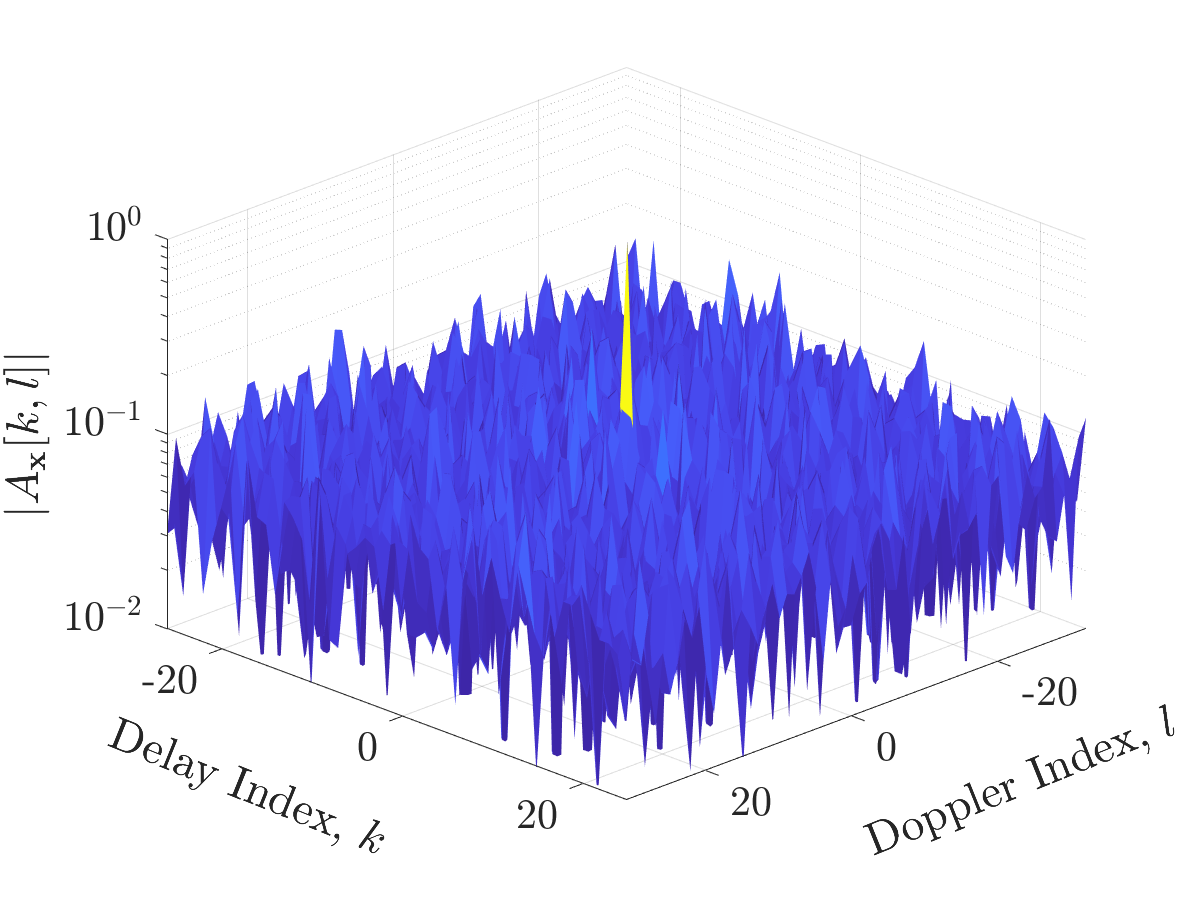}
    \caption{AFDM}
        \label{fig:crossamb_data3}
    \end{subfigure}
    \begin{subfigure}{0.24\linewidth}
    \includegraphics[width=\textwidth]{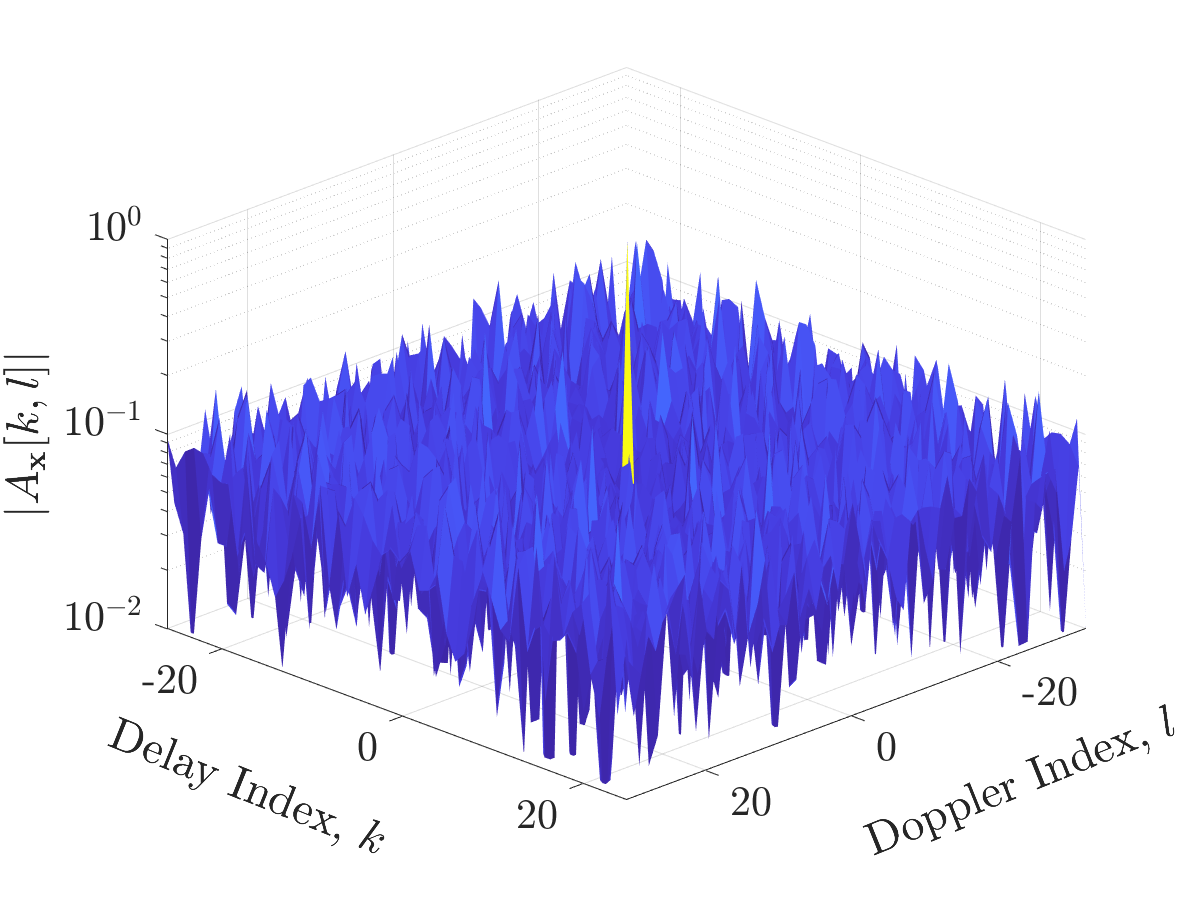}
    \caption{OFDM}
        \label{fig:crossamb_data4}
    \end{subfigure}
    \caption{Data-based self-ambiguity function $\mathbf{A}_{\mathbf{x}}[k,l]$ approximates an ideal ``thumbtack'' in expected value for any basis $\boldsymbol{\phi}$. The information symbols are drawn from $4$-QAM in (a)-(d) and from $16$-QAM in (e)-(h).}
    \label{fig:crossamb_data}
\end{figure*}

\begin{IEEEproof}
    The self-ambiguity $\mathbf{A}_{\mathbf{x}}[k,l]$ of $\mathbf{x}$ in~\eqref{eq:prelim3} is:
    \begin{align}
        \label{eq:sys4}
        \mathbf{A}_{\mathbf{x}}[k,l] &= \frac{1}{MN} \sum_{n = 0}^{MN-1} \mathbf{x}[n] \mathbf{x}^{*}\big[(n-k)_{{}_{MN}}\big] e^{-\frac{j2\pi}{MN}l(n-k)} \nonumber \\
        &= \sum_{i = 0}^{MN-1} \sum_{j = 0}^{MN-1} \mathbf{s}[i] \mathbf{s}^{*}[j] \mathbf{A}_{\boldsymbol{\phi}_{i},\boldsymbol{\phi}_{j}}[k,l] \nonumber \\
        &= \sum_{i = 0}^{MN-1} |\mathbf{s}[i]|^{2} \mathbf{A}_{\boldsymbol{\phi}_{i}}[k,l] \nonumber \\ &+ \sum_{i = 0}^{MN-1} \sum_{j \neq i} \mathbf{s}[i] \mathbf{s}^{*}[j] \mathbf{A}_{\boldsymbol{\phi}_{i},\boldsymbol{\phi}_{j}}[k,l].
    \end{align}
    
    We assume the data symbols $\mathbf{s}[i] \in \mathcal{A}$ are drawn i.i.d. from a \emph{unit energy}, \emph{zero-mean} constellation $\mathcal{A}$; hence, $|\mathbf{s}[i]|^{2} = 1$ and $\mathbb{E}\big[\mathbf{s}[i] \mathbf{s}^{*}[j]\big] = \delta[i-j]$. The expected value of~\eqref{eq:sys4} is thus:
    \begin{align}
        \label{eq:sys5}
        \mathbb{E}\big[\mathbf{A}_{\mathbf{x}}[k,l]\big] &= \sum_{i = 0}^{MN-1} \underbrace{\mathbb{E}\big[|\mathbf{s}[i]|^{2}\big]}_{1} \mathbf{A}_{\boldsymbol{\phi}_{i}}[k,l] \nonumber \\ &+ \sum_{i = 0}^{MN-1} \sum_{j \neq i} \underbrace{\mathbb{E}\big[\mathbf{s}[i] \mathbf{s}^{*}[j]\big]}_{\delta[i-j] = 0~\text{since}~j \neq i} \mathbf{A}_{\boldsymbol{\phi}_{i},\boldsymbol{\phi}_{j}}[k,l] \nonumber \\
        &= \sum_{i = 0}^{MN-1} \mathbf{A}_{\boldsymbol{\phi}_{i}}[k,l] \nonumber \\
        &= \sum_{n = 0}^{MN-1} e^{-\frac{j2\pi}{MN}l(n-k)} \nonumber \\ &\times \frac{1}{MN} \sum_{i = 0}^{MN-1} \boldsymbol{\phi}_{i}[n] \boldsymbol{\phi}_{i}^{*}\big[(n-k)_{{}_{MN}}\big],
    \end{align}
    where the summation over $i$ evaluates to $\mathds{1}\big\{k \equiv 0 \bmod{MN}\big\}$ since the basis $\boldsymbol{\phi}$ is orthonormal. Therefore, we obtain:
    \begin{align}
        \label{eq:sys6}
        \mathbb{E}\big[\mathbf{A}_{\mathbf{x}}[k,l]\big] &= \sum_{n = 0}^{MN-1} \mathds{1}\big\{k \equiv 0 \bmod{MN}\big\} e^{-\frac{j2\pi}{MN}l(n-k)} \nonumber \\
        &= \mathds{1}\big\{k \equiv 0 \bmod{MN}\big\} \frac{1}{MN} \sum_{n = 0}^{MN-1} e^{-\frac{j2\pi}{MN}ln} \nonumber \\
        &= \mathds{1}\big\{k,l \equiv 0 \bmod{MN}\big\},
    \end{align}
    since the final expression follows from Identity~\ref{idty:sumrootsofunity}.
    
    Therefore, for any orthonormal basis $\boldsymbol{\phi}$, we have:
    \begin{align}
        \label{eq:sys7}
        \mathbb{E}\big[\widehat{\mathbf{h}}[k,l]\big] &= \mathbf{h}[k,l] *_{\sigma_{{}_{d}}} \mathds{1}\big\{k,l \equiv 0 \bmod{MN}\big\} \nonumber \\ 
        &= \mathbf{h}[k,l],
    \end{align}
    assuming $\mathbf{h}[k,l]$ is supported within $0 \leq k,l \leq (MN-1)$.
\end{IEEEproof}

Fig.~\ref{fig:crossamb_data} illustrates the above result by plotting the magnitude of the self-ambiguity function $\mathbf{A}_{\mathbf{x}}[k,l]$ for $4$-QAM and $16$-QAM data modulated on different bases $\boldsymbol{\phi}$. As suggested by Theorem~\ref{thm:data_amb}, $\mathbf{A}_{\mathbf{x}}[k,l]$ approximates an ideal ``thumbtack'' in expected value, with cross-interactions between data symbols (second summation in~\eqref{eq:sys4}) resulting in a non-zero noise floor.

\begin{figure*}
    \centering
    \begin{subfigure}{0.47\linewidth}
    \includegraphics[width=\textwidth]{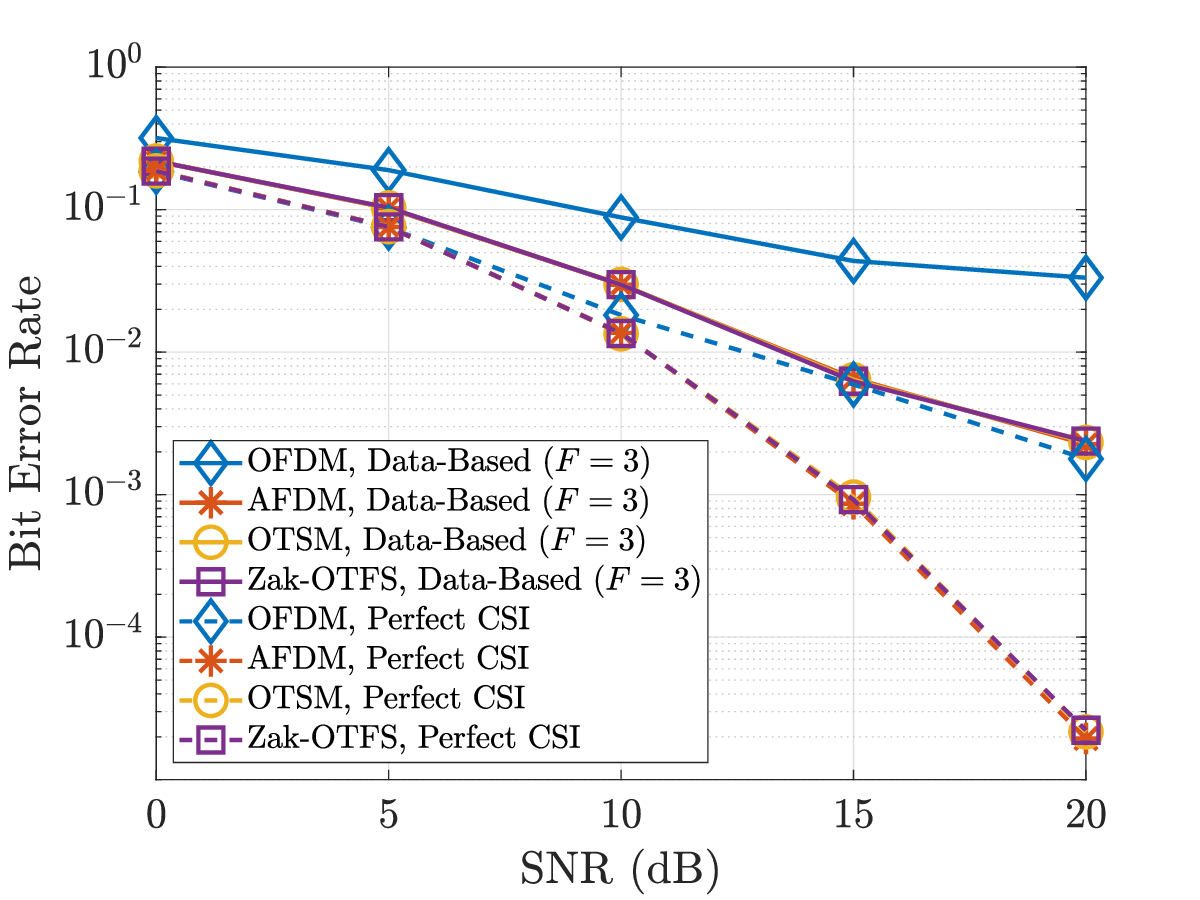}
    \caption{BER comparison of different modulations.}
        \label{fig:ber_vs_wvf}
    \end{subfigure}
    \begin{subfigure}{0.47\linewidth}
        \includegraphics[width=\textwidth]{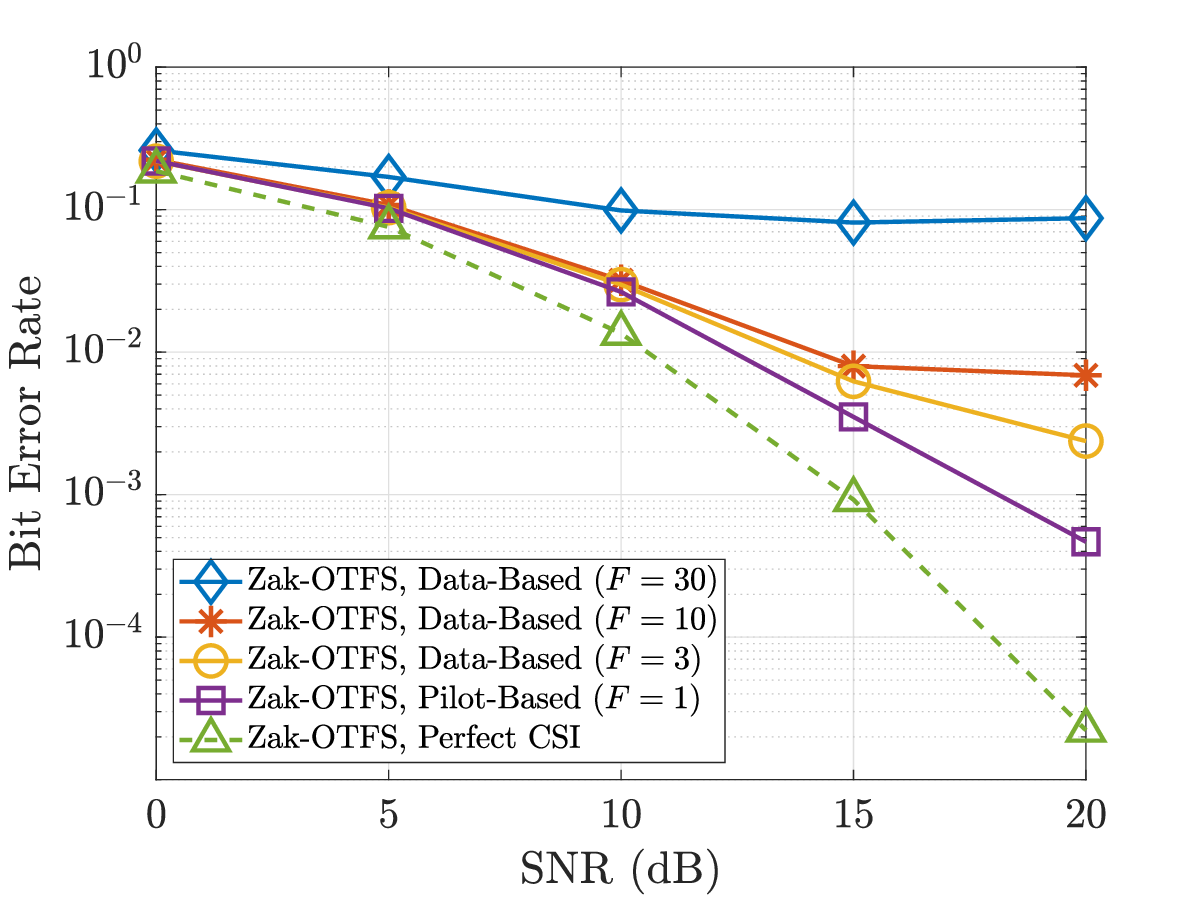}
    \caption{Data-based Zak-OTFS BER for different values of $F$.}
        \label{fig:ber_zak_vs_f}
    \end{subfigure}

    \vspace{0.5em}

    \begin{subfigure}{0.47\linewidth}
    \includegraphics[width=\textwidth]{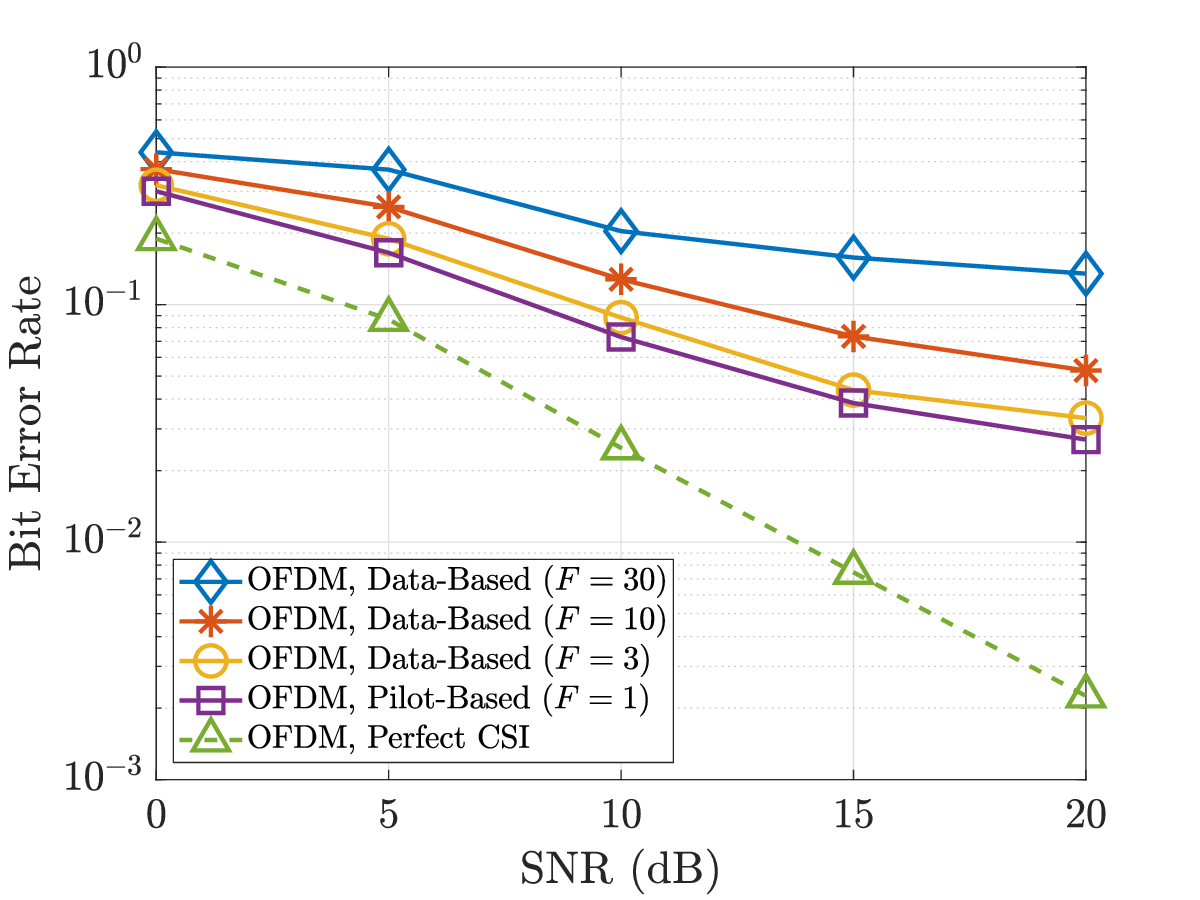}
    \caption{Data-based OFDM BER for different values of $F$.}
        \label{fig:ber_ofdm_vs_f}
    \end{subfigure}
    \begin{subfigure}{0.47\linewidth}
        \includegraphics[width=\textwidth]{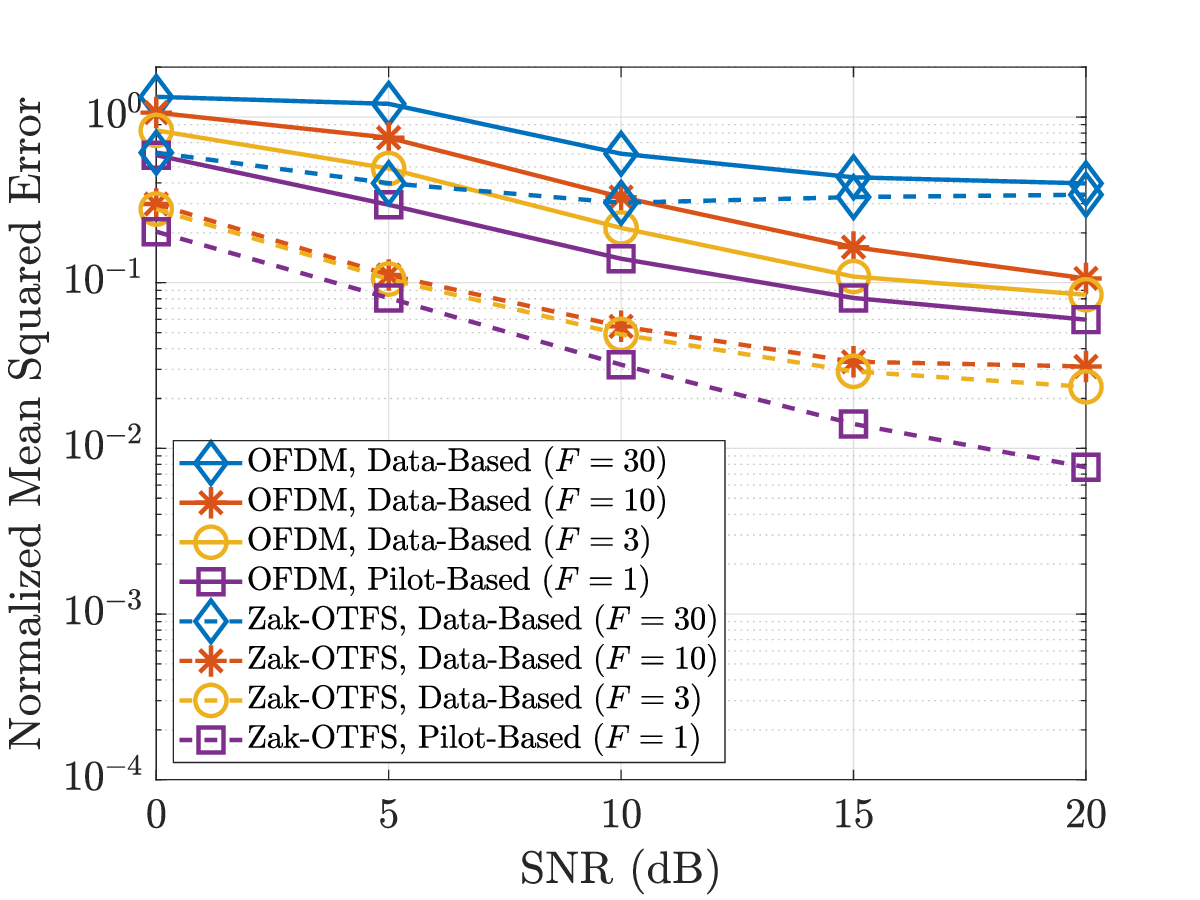}
    \caption{Data-based DD channel estimation NMSE.}
        \label{fig:nmse_vs_wvf_f}
    \end{subfigure}    

    \vspace{0.5em}

    \begin{subfigure}{0.47\linewidth}
        \includegraphics[width=\textwidth]{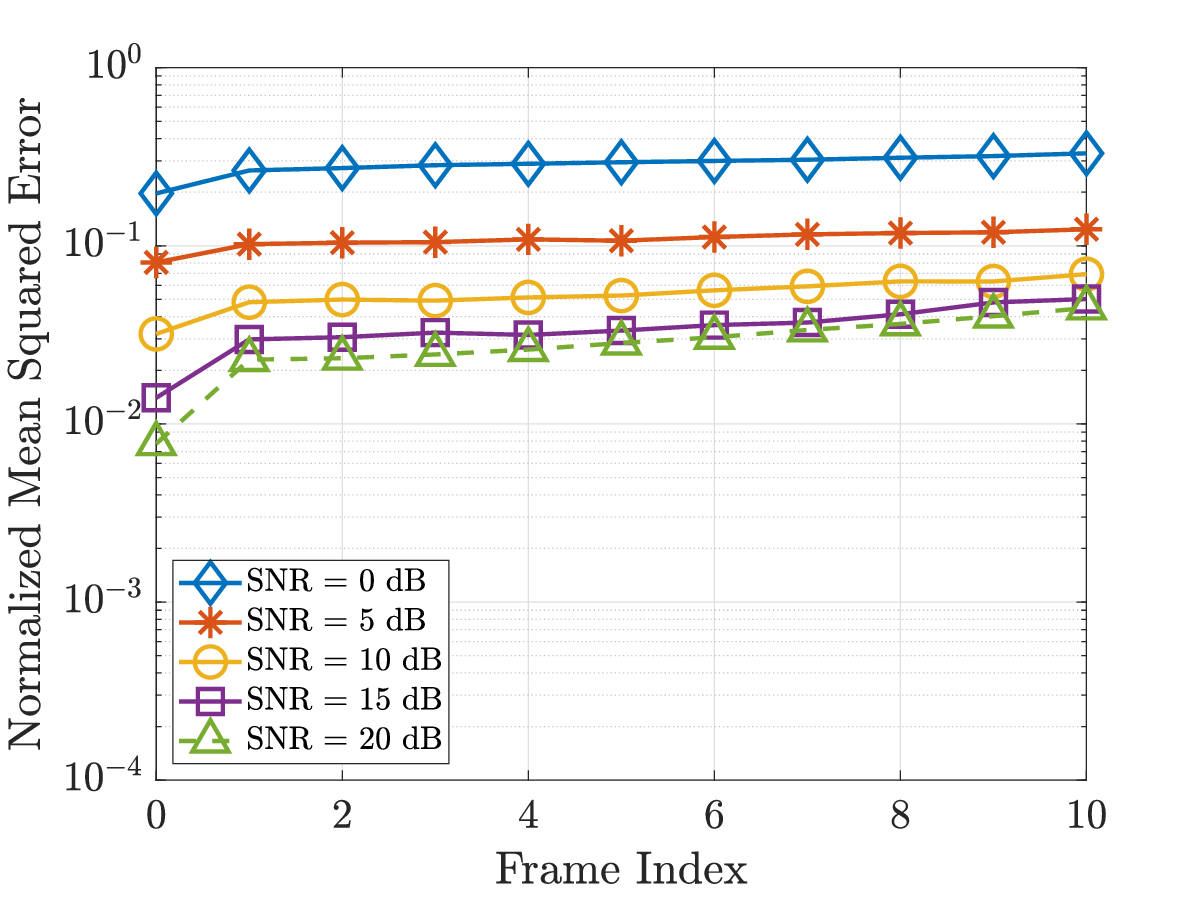}
    \caption{Zak-OTFS channel NMSE across pilot \& $F = 10$ data frames.}
        \label{fig:error_prop_nmse}
    \end{subfigure}
    \begin{subfigure}{0.47\linewidth}
    \includegraphics[width=\textwidth]{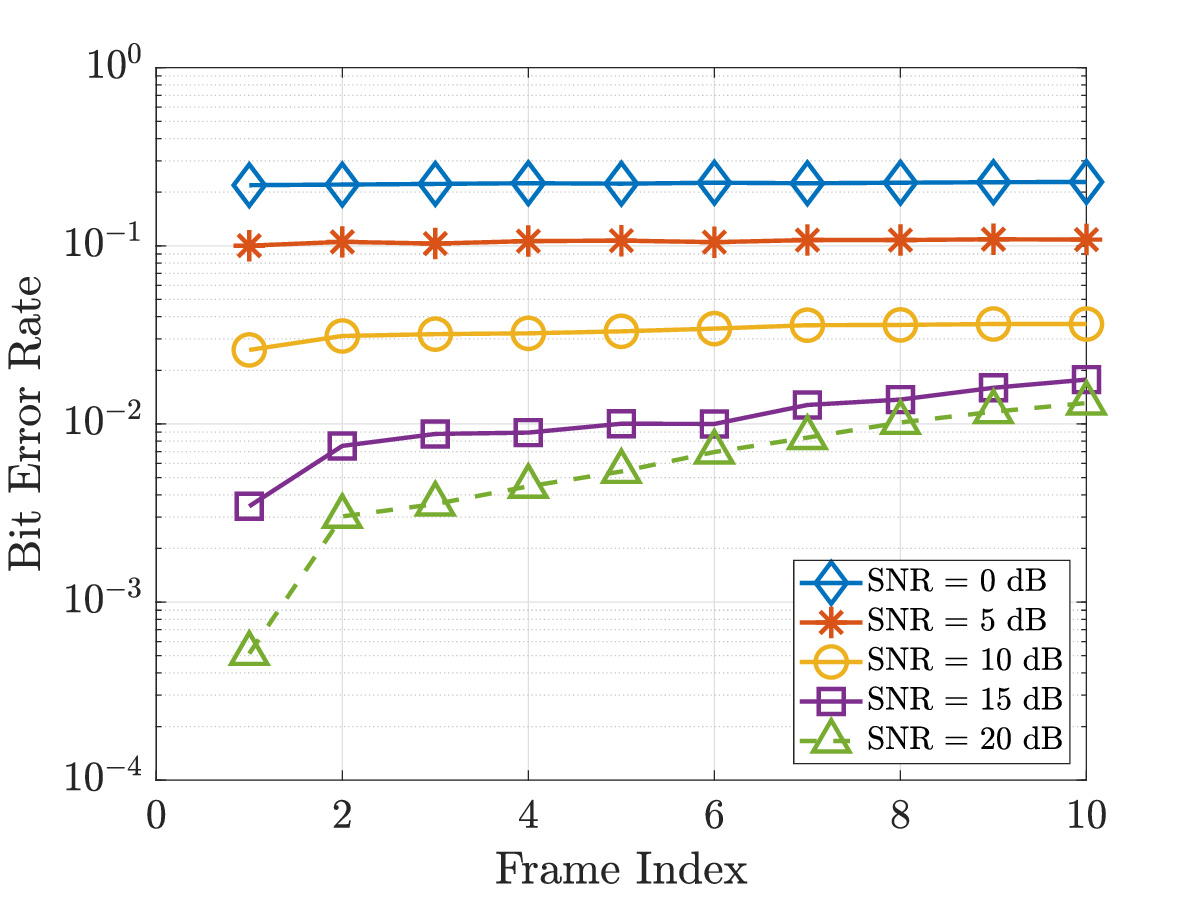}
    \caption{Data-based Zak-OTFS BER across $F = 10$ data frames.}
        \label{fig:error_prop_ber}
    \end{subfigure}
    
    \caption{Performance of both systems in Fig.~\ref{fig:overview} for uncoded $4$-QAM data modulated using various modulation schemes. (a) Bit error rate (BER) remains similar for Zak-OTFS, AFDM and OTSM, with significant gains over OFDM due to its non-predictability~\cite{Mehrotra2026_wvfcomp} and limitations of one-tap equalization. (b)-(c) BER of data-based systems degrade with increasing number of data frames $F$ due to error propagation. (d) Normalized mean squared error (NMSE) of data-based DD channel estimation increases as a function of $F$ and saturates at high SNR due to the data-based self-ambiguity noise floor (Fig.~\ref{fig:crossamb_data}). (e)-(f) Error propagation across $F = 10$ data frames. Accurate DD channel estimates from the pilot frame $f' = 0$ result in good BER performance in the initial data frame $f' = 1$. However, symbol detection errors in frame $f' = 1$ and the non-zero noise floor in the data-based self-ambiguity function in Fig.~\ref{fig:crossamb_data} degrade subsequent channel estimation and data detection performance.}
    \label{fig:ber_nmse_wvf_f}
\end{figure*}

\subsection{Necessary Conditions for Practical Implementation}
\label{subsec:diff_comm_feasibility}

To enable data-based DD channel estimation as per Theorem~\ref{thm:data_amb}, we consider the frame structure illustrated in Fig.~\ref{fig:overview}(\subref{fig:overview2}) with $F$ data frames in between pilot frames. Let $1 \leq f' \leq F$ sequentially index the $F$ data frames. DD channel estimates obtained from the pilot frame $f' = 0$ are used for equalizing and decoding the data symbols in frame $f' = 1$. Subsequently, a data-based estimate of the DD channel is obtained using the decoded data symbols, and the data-based DD channel estimate from frame $f' = 1$ is used to decode the data symbols in frame $f' = 2$. This process is repeated for all remaining data frames, and the entire procedure restarts on the transmission of another pilot frame after $F$ data frames.

% Given the sequential nature of the approach, it is crucial that the initial pilot-based DD channel estimates are accurate to ensure near-optimal data decoding performance in frame $f' = 1$. This requires: (i) channel coherence time spanning at least two frame durations, and (ii) frame bandwidth within the channel coherence bandwidth, i.e.,
The sequential nature of our approach requires accurate initial pilot-based DD channel estimates for near-optimal data decoding performance in frame $f' = 1$. This requires: (i) channel coherence time spanning at least two frame durations, and (ii) frame bandwidth within the channel coherence bandwidth:
\begin{align}
    \label{eq:cond}
    2T = \frac{2N}{\Delta f} \leq T_{\mathsf{c}} = \frac{1}{\nu_{\max}} &\implies \Delta f \geq 2N \nu_{\max}, \nonumber \\
    B = M \Delta f \leq B_{\mathsf{c}} = \frac{1}{\tau_{\max}} &\implies \Delta f \leq \frac{1}{M\tau_{\max}}.
\end{align}
% where all variables are defined as per Section~\ref{sec:prelim}.

\section{Numerical Results}
\label{sec:results}

\subsection{Simulation Configuration}
\label{subsec:sim_config}

We conduct numerical simulations using a 3GPP-compliant $P=6$ path Vehicular-A (Veh-A) channel model~\cite{veh_a}, whose power-delay profile is shown in Table~\ref{tab:veh_a}. The Doppler of each path is simulated as $\nu_i = \nu_{\max}\cos(\theta_i)$, with $\theta_i$ uniformly distributed in $[-\pi, \pi)$ and $\nu_{\max} = 815$ Hz denoting the maximum channel Doppler spread\footnote{Our channel model represents propagation environments with \textit{fractional} delay and Doppler shifts since the path delays $\tau_i$ in Table~\ref{tab:veh_a} and Doppler shifts $\nu_i$ are non-integer multiples of the respective resolutions $\nicefrac{1}{B}$ and $\nicefrac{1}{T}$. Further, at $3$ GHz carrier frequency, $\nu_{\max} = 815$ Hz corresponds to a highly dynamic scenario with maximum speed $81.5$ m/s, e.g., a high-speed train.}. To satisfy the necessary conditions in~\eqref{eq:cond}, we consider parameters\footnote{There is no per-symbol cyclic prefix (CP) in Zak-OTFS, AFDM, OTSM. In OFDM, we simulate the channel after CP removal, where CP $> \tau_{\max}$.}: $M=13, N=16$, $\Delta f = 30$ kHz, for which $B = 0.39$ MHz, $T = 0.533$ ms, such that $2N\nu_{\max} = 26.08~\text{kHz} \leq \Delta f = 30~\text{kHz} \leq \frac{1}{M \tau_{\max}} = 30.647$ kHz. In every pilot frame, we generate a random DD channel realization with random per-path Doppler and per-path channel gain $h_i = \alpha_i e^{j\psi_i}$, where the relative power $|\alpha_i|^2$ is as per Table~\ref{tab:veh_a} and $\psi_i$ uniformly distributed in $[-\pi, \pi)$. The channel is subsequently evolved forward across $F$ data frames as $\tau_i(f') = \tau_i + \frac{\nu_i}{f_c}f'T$ and $h_i(f') = h_i \big(\frac{1+c\tau_i}{1+c\tau_i(f')}\big)$ for center frequency $f_c = 2.4$ GHz and speed of light $c$. To generate the channel spreading function $\mathbf{h}(\tau,\nu)$, we consider a Gaussian-sinc pulse shape $\mathbf{w}(\tau,\nu)$ in~\eqref{eq:prelim1a}, see~\cite{Chockalingam2025_gs,Mehrotra2026_iota} for more details.

We simulate both systems depicted in Fig.~\ref{fig:overview} with uncoded $4$-QAM and $16$-QAM data (which satisfy the condition in Theorem~\ref{thm:data_amb}) modulated using OFDM, AFDM, OTSM and Zak-OTFS. We assume no data symbols in the pilot frames\footnote{In the pilot frame, $\mathbf{x} = \boldsymbol{\phi}_{i}$ for some $i \in \mathbb{Z}_{MN}$ for AFDM, OTSM and Zak-OTFS, whereas for OFDM, $\mathbf{x}$ is per~\eqref{eq:prelim4} with all $MN$ symbols $\mathbf{s}$ known.} for all four modulations, with equal signal-to-noise ratio (SNR)\footnote{For Gaussian noise with known variance $\sigma^2$, $\text{SNR} = \nicefrac{\Vert \mathbf{y} \Vert_{2}^{2}}{(MN\sigma^2)}$.} for the pilot and data frames. We perform data detection using the minimum mean squared error (MMSE) estimator\footnote{Matrix $\mathbf{G}$ in~\eqref{eq:prelim4} is estimated in AFDM, OTSM, Zak-OTFS, whereas for OFDM, transfer-domain channel diagonals are estimated (one-tap equalizer).}~\cite{Tse2005} with known noise variance and hard symbol decisions.

\begin{table}[!t]
    \centering
    \caption{Power-delay profile of Veh-A channel model}
    \begin{tabular}{|c|c|c|c|c|c|c|}
         \hline
         Path index $i$ & $1$ & $2$ & $3$ & $4$ & $5$ & $6$ \\
         \hline
         Delay $\tau_i (\mu s)$ & $0$ & $0.31$ & $0.71$ & $1.09$ & $1.73$ & $2.51$ \\
         \hline
         Relative power $|\alpha_i|^2$ (dB) & $0$ & $-1$ & $-9$ & $-10$ & $-15$ & $-20$ \\
         \hline
    \end{tabular}
    \label{tab:veh_a}
\end{table}

\subsection{Overall System Performance}
\label{subsec:sim_data_detecn}

\begin{figure}
    \centering
    \includegraphics[width=0.95\linewidth]{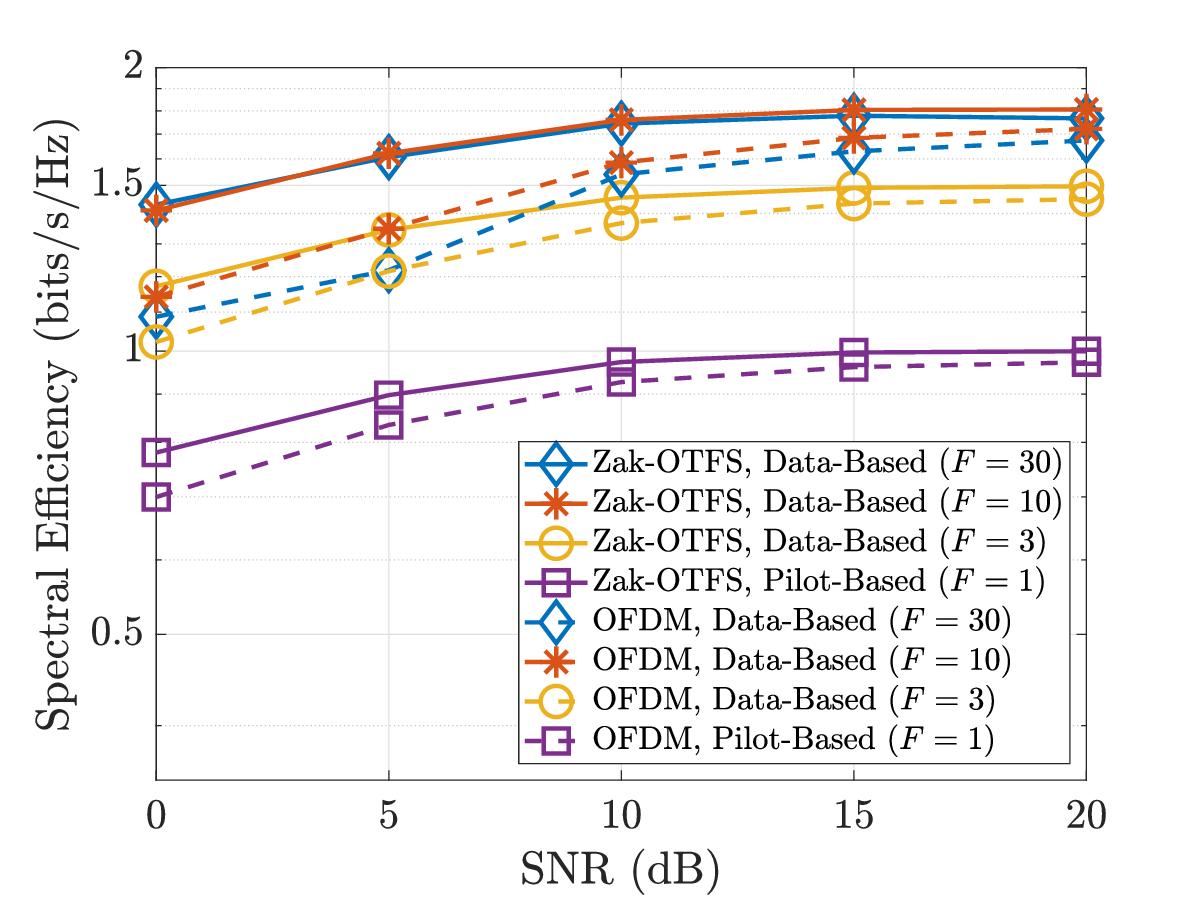}
    \caption{Data-based systems achieve higher spectral efficiency than pilot-based systems due to smaller pilot overhead.}
    \label{fig:se}
\end{figure}

\begin{figure*}
    \centering
    \begin{subfigure}{0.46\linewidth}
    \includegraphics[width=\textwidth]{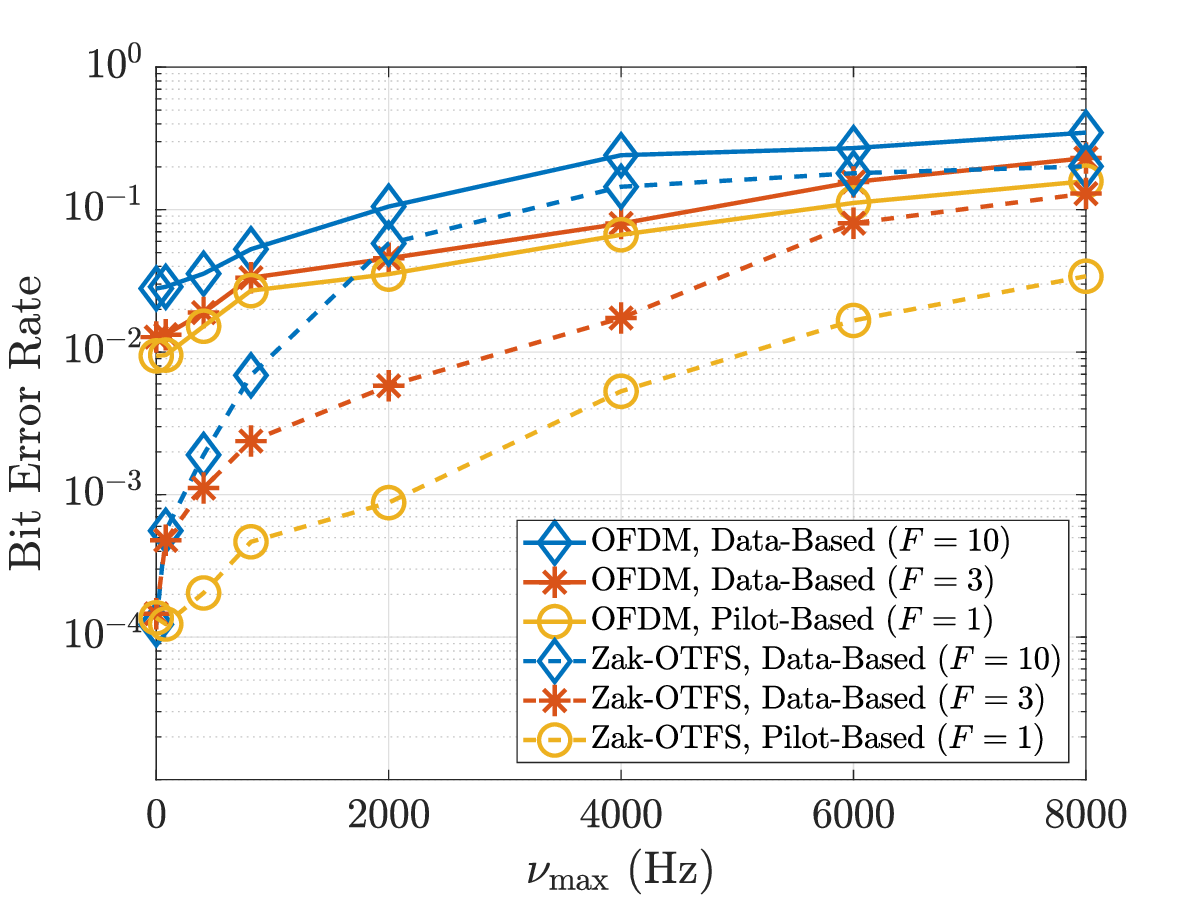}
    \caption{Bit error rate (BER).}
        \label{fig:ber_vs_numax}
    \end{subfigure}
    \begin{subfigure}{0.46\linewidth}
    \includegraphics[width=\textwidth]{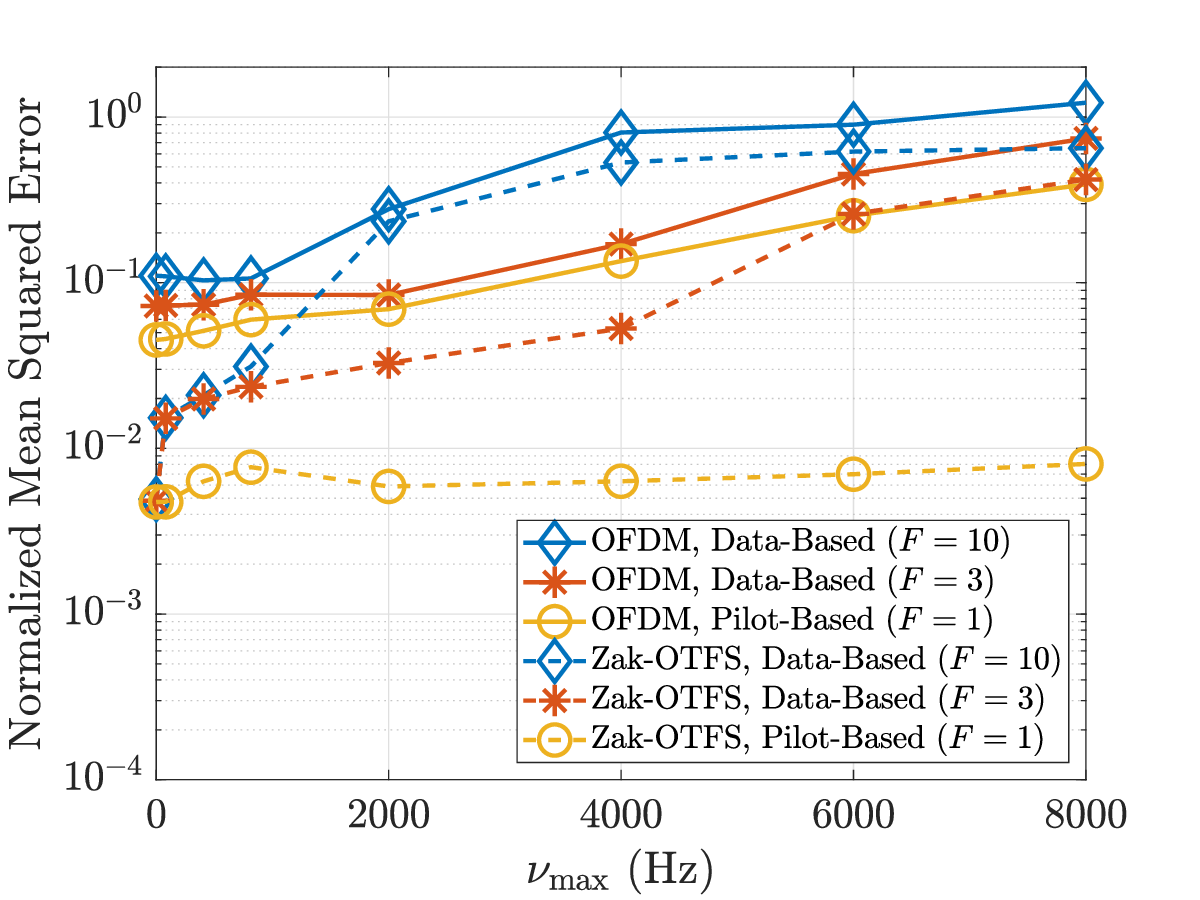}
    \caption{Normalized mean squared error (NMSE).}
        \label{fig:nmse_vs_numax}
    \end{subfigure}
    \caption{Performance of data-based systems degrades for $\nu_{\max} > 937.5$ Hz where the necessary conditions in~\eqref{eq:cond} do not hold.}
    \label{fig:ber_nmse_vs_numax}
\end{figure*}

\begin{figure}
    \centering
    \includegraphics[width=0.95\linewidth]{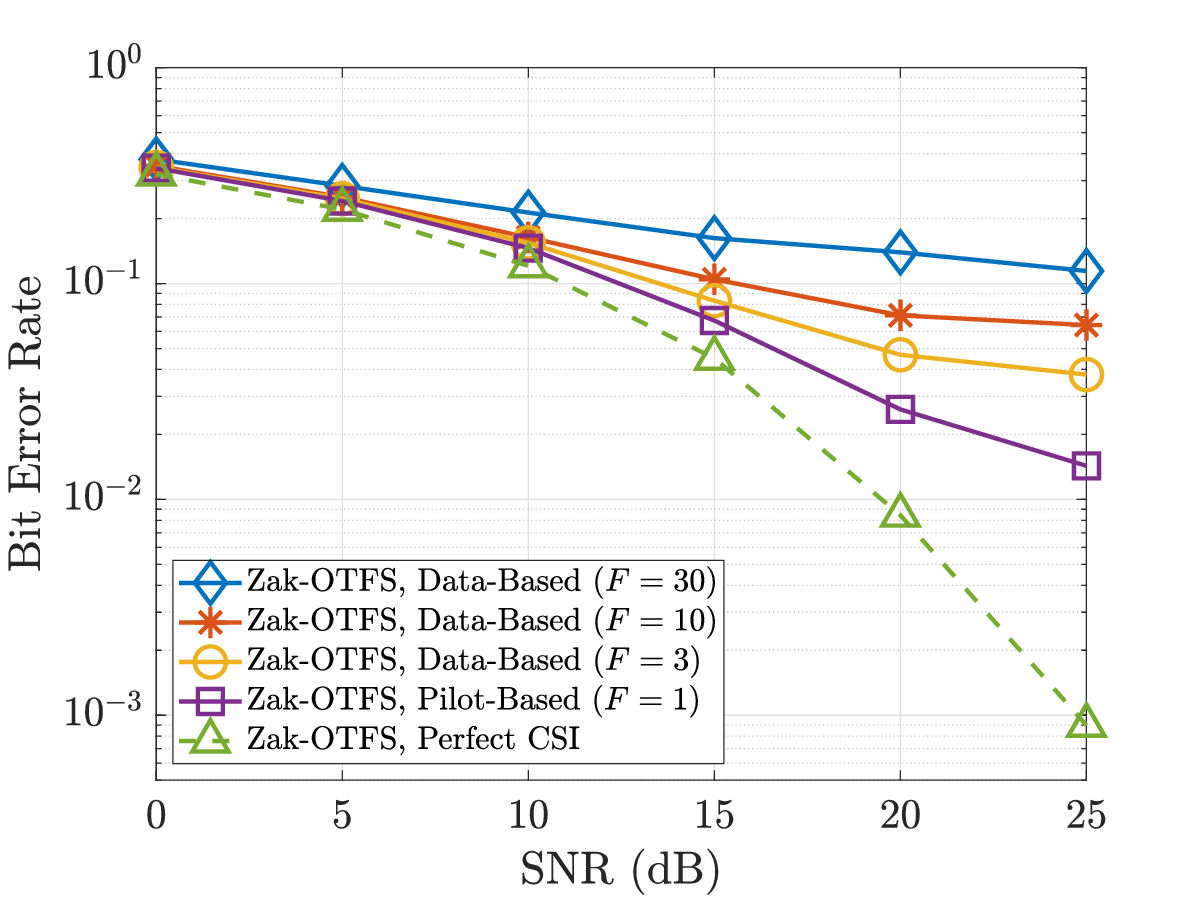}
    \caption{Data-based DD channel estimation with $16$-QAM.}
    \label{fig:16qam_vs_f}
\end{figure}

Fig.~\ref{fig:ber_nmse_wvf_f} compares the performance of the pilot-based and data-based systems from Fig.~\ref{fig:overview} with uncoded $4$-QAM data. 

Fig.~\ref{fig:ber_nmse_wvf_f}(\subref{fig:ber_vs_wvf}) shows that the bit error rate\footnote{We simulate $\geq 2500$ realizations per SNR until $200$ errors are counted.} (BER) of predictable modulation schemes (AFDM, OTSM and Zak-OTFS)~\cite{Mehrotra2026_wvfcomp} is similar\footnote{Due to similar performance, subsequent results only consider Zak-OTFS.} with gains over OFDM since the latter is not predictable, hence has poor performance even with perfect channel state information (CSI). Data-based AFDM / OTSM / Zak-OTFS systems with $F = 3$ offer similar performance as OFDM with perfect CSI. Data-based OFDM has further degraded performance due to poor initial pilot-based channel estimates as a result of mobility-caused inter-carrier interference that cannot be estimated in one-tap equalization.

Figs.~\ref{fig:ber_nmse_wvf_f}(\subref{fig:ber_zak_vs_f})-(\subref{fig:ber_ofdm_vs_f}) illustrate the BER for data-based Zak-OTFS and OFDM systems. The performance degrades as the number of data frames $F$ increases due to increased error propagation -- errors in symbol detection degrade data-based DD channel estimation, which degrades symbol detection, and so on.

Fig.~\ref{fig:ber_nmse_wvf_f}(\subref{fig:nmse_vs_wvf_f}) illustrates the normalized mean squared error (NMSE) for DD channel estimation, defined as $\text{NMSE} = \frac{\sum_{k,l}|\widehat{\mathbf{h}}_{\mathrm{eff}}[k,l] - \mathbf{h}_{\mathrm{eff}}[k,l]|^{2}}{\sum_{k,l}|\mathbf{h}_{\mathrm{eff}}[k,l]|^{2}}$. The NMSE increases as a function of $F$ and saturates at high SNR due to the noise floor in the self-ambiguity function $\mathbf{A}_{\mathbf{x}}[k,l]$ in Fig.~\ref{fig:crossamb_data}.

Figs.~\ref{fig:ber_nmse_wvf_f}(\subref{fig:error_prop_nmse})-(\subref{fig:error_prop_ber}) illustrate error propagation across $F = 10$ data frames due to the sequential nature of the proposed approach. While initial pilot-based channel estimates result in accurate data detection in frame $f' = 1$, symbol detection errors and the non-zero noise floor in the data-based self-ambiguity function in Fig.~\ref{fig:crossamb_data} degrade data-based channel estimation NMSE in frame $f' = 1$. The degraded channel estimate degrades data detection performance in the next frame $f' = 2$, which further degrades current-frame data-based NMSE, and so on.

\subsection{Spectral Efficiency Comparison}
\label{subsec:sim_se}

Fig.~\ref{fig:se} compares the spectral efficiency, defined as $\text{SE} = \big(1-O)(1-\text{BER})\frac{MN\log_{2}{|\mathcal{A}|}}{BT}$ bits/s/Hz, where $O$ denotes the pilot overhead ($\nicefrac{1}{2}$ for Fig.~\ref{fig:overview}(\subref{fig:overview1}) and $\nicefrac{1}{(F+1)}$ for Fig.~\ref{fig:overview}(\subref{fig:overview2})), for $4$-QAM. For data-based Zak-OTFS, the degradation in BER from $F = 10$ to $F = 30$ in Fig.~\ref{fig:ber_nmse_wvf_f}(\subref{fig:ber_zak_vs_f}) is offset by the reduction in pilot overhead for $F = 30$, resulting in both systems providing the best possible SE of $\sim 1.8$ bits/s/Hz. However, the larger pilot overhead for small $F$ values does not compensate for the improvement in BER, resulting in poor SE. Similar conclusions follow for OFDM, however, its poorer BER in Fig.~\ref{fig:ber_nmse_wvf_f}(\subref{fig:ber_vs_wvf}) compared to Zak-OTFS reduces the SE.

\subsection{Impact of Channel Mobility}
\label{subsec:sim_numax}

Fig.~\ref{fig:ber_nmse_vs_numax} illustrates the impact of the channel Doppler spread $\nu_{\max}$ on the $4$-QAM system performance. For the choice of parameters in Section~\ref{subsec:sim_config}, the necessary conditions in~\eqref{eq:cond} are satisfied when $\nu_{\max} \leq \frac{\Delta f}{2N} = 937.5$ Hz. This is consistent with Fig.~\ref{fig:ber_nmse_vs_numax} where the performance degrades significantly for both Zak-OTFS and OFDM systems beyond $\nu_{\max} > 937.5$ Hz. However, predictable modulations such as Zak-OTFS are more resilient to higher Doppler spreads since they do not suffer from inter-carrier interference, unlike OFDM.

\subsection{Performance with $16$-QAM}
\label{subsec:sim_16qam}

Fig.~\ref{fig:16qam_vs_f} illustrates the data-based Zak-OTFS system performance with $16$-QAM, showing the feasibility of data-based DD channel estimation. As suggested by Theorem~\ref{thm:data_amb}, data-based DD channel estimation remains feasible with all zero-mean, unit average energy constellations, including $16$-QAM, albeit with poorer BER than $4$-QAM due to the smaller minimum distance between constellation points.

\section{Conclusion}
\label{sec:conclusion}

In this paper, we proposed a data-based DD channel estimation approach to reduce pilot overhead and increase spectral efficiency. The proposed approach is applicable to any modulation scheme provided the information symbols are drawn from a zero-mean, unit average energy constellation. Numerical results demonstrated $\sim 1.8 \times$ improvement in uncoded spectral efficiency over conventional pilot-based approaches. Future work will consider coding, constellation shaping, and generalize the approach to multi-antenna, multi-user systems.

% In this paper, we proposed a data-based DD channel estimation approach to reduce pilot overhead and increase spectral efficiency. The proposed approach is applicable to any modulation scheme provided the information symbols are drawn from a unit energy, zero-mean constellation. Numerical results with uncoded \textcolor{blue}{$4$-QAM} demonstrated $\sim 1.8 \times$ improvement in spectral efficiency over conventional pilot-based approaches. Future work will consider coding, constellation shaping, and turbo-based equalization, and also pursue generalizations of the approach to multi-antenna, multi-user systems.

\bibliographystyle{IEEEtran}
\bibliography{references}
\end{document}